\documentclass[12pt]{article}
\pdfoutput=1
\usepackage[nosort]{cite}

\usepackage{epsfig}
\usepackage{amsfonts}
\usepackage{amscd}
\usepackage{latexsym}
\usepackage{amsmath,amssymb}
\usepackage{verbatim}
\usepackage{setspace}
\usepackage[dvipsnames]{xcolor}
\usepackage{fancyhdr}
\usepackage{cite}

\usepackage{tikz}
\usepackage{tikz-cd}
\usepackage{tikzit}

\tikzstyle{pointoperator}=[fill=black, draw=none, shape=circle, minimum size=.1 cm, inner sep=0 pt]

\tikzstyle{dashedline}=[dashed, -, thick]
\tikzstyle{blueline}=[-, draw=blue, thick]
\tikzstyle{arrowline}=[<-, thick]
\tikzstyle{normalline}=[-, thick, draw={rgb,255: red,171; green,171; blue,171}]
\tikzstyle{fillline}=[-, thick, fill=red, fill opacity=.5]
\tikzstyle{bluefillline}=[-, fill opacity=.5, fill=blue, thick]
\tikzstyle{purplefillline}=[-, fill opacity=.5, fill=purple, thick]
\tikzstyle{bluedashedline}=[-, draw=blue, thick, dashed]

\usepackage{slashed}
\usepackage{multirow}
\usetikzlibrary{calc}
\usetikzlibrary{topaths}
\usetikzlibrary{decorations}
\usetikzlibrary{decorations.pathmorphing}
\usetikzlibrary{arrows,decorations.markings,cd}
\usetikzlibrary{calc,arrows,cd,decorations.markings,snakes}
\tikzset{
->-/.style args={#1rotate#2}{decoration={markings, mark=at position #1 with {\arrow[scale=1.5,rotate = #2 ]{stealth}}}, postaction={decorate}}
}
\usetikzlibrary{shapes.geometric}
\usetikzlibrary{knots}
\usepackage{tikz-3dplot}

\tikzset{snake it/.style={decorate, decoration=snake}}

\usepackage{draft}

\usepackage{graphicx,color,subcaption}
\usepackage{cite}
\usepackage{mciteplus}
\usepackage{skak}
\usepackage{bbm}
\usepackage[english]{babel}
\usepackage{amsthm}
\usepackage{soul}

\numberwithin{equation}{section}

\newtheorem{lemma}{Lemma}

\tikzset{
    mid arrow/.style={
        postaction={
            decorate,
            decoration={
                markings,
                mark=at position 0.6 with {\arrow[scale=1.0]{stealth}}
            }
        }
    },
    inline tikz/.style={
        baseline={([yshift=-.5ex]current bounding box.center)}
    }
}

\usepackage{hyperref}
\hypersetup{
    colorlinks=true,
    linkcolor=Maroon,
    citecolor=MidnightBlue,
    urlcolor=MidnightBlue,
    linktoc=page
}

\begin{document}

\begin{titlepage}

\title{Gauss Law, Monodromy Defect, Magnetic Lattice Translation, and Lattice Duality}

\author{Pengcheng Wei$^1$ and Yunqin Zheng$^{2}$}

        \address{${}^{1}$ Department of Physics, Nanjing University, Nanjing, Jiangsu 210093, China\\\bigskip 
        ${}^{2}$  Kavli Institute for Theoretical Sciences and School of Quantum, \\University of Chinese Academy of Sciences, Beijing, 100190, China}

\abstract
In a Hamiltonian lattice gauge theory, the physical Hilbert space is fixed by the Gauss law, by choosing an eigenspace of the Gauss law operator which generates the gauge transformation. Choosing a non-trivial eigenspace is referred to as the modified Gauss law. In 2+1d dynamical $\mathrm{U(1)}$ pure gauge theory, we show that these physical Hilbert spaces are defect Hilbert spaces associated with non-topological monodromy defects for monopole operators and topological defects for $\mathrm{U(1)}$ 1-form symmetry, inserted along the time direction. We also construct the corresponding lattice translation operators. Under a nontrivially modified Gauss law, these translations obey a magnetic translation algebra. Finally, we discuss how the modified Gauss laws are manifested in 2+1d compact bosons, via exact duality on the lattice.

\end{titlepage}

\eject

\setcounter{tocdepth}{3}

\tableofcontents

\section{Introduction}

In a quantum field theory or lattice model with flat or dynamical gauge fields, physical observables must be gauge invariant. In a Hamiltonian lattice model, one typically begins with an extended tensor-product Hilbert space $\mathcal{H}$. Gauge transformations acting on $\mathcal{H}$ are generated by the Gauss law operator $G[\zeta(x)]$, where $\zeta(x)$ is a spatially dependent gauge parameter. The physical Hilbert space is defined as the subspace of $\mathcal{H}$ satisfying
\begin{equation}\label{eq:intro the Gauss law}
    G[\zeta(x)] = 1 \quad \text{for all } \zeta(x).
\end{equation}
The constraint \eqref{eq:intro the Gauss law} is referred to as the Gauss law.

In fact, the notion of Gauss laws is suggested to be generalized to any constraint imposed on the Hilbert space, not necessarily associated with the gauge group as an analog from the continuum. See more discussions in footnote \ref{fn:gauss laws and constraints}.

However, we are not strictly required to define the physical Hilbert space as the eigenspace of $G[\zeta(x)]$ with eigenvalue $1$. Instead, we can choose a more general eigenspace governed by a modified Gauss law:
\begin{equation}\label{eq:intro the modified Gauss law}
    G[\zeta(x)] = \chi[\zeta(x)],
\end{equation}
where $\chi$ takes values in $\mathrm{U}(1)$ since $G$ is unitary. The modified Gauss law \eqref{eq:intro the modified Gauss law} defines a distinct physical Hilbert space, denoted by $\mathcal{H}_\chi$.

A natural question then arises: \emph{what is the interpretation of $\mathcal{H}_\chi$ for a nontrivial $\chi\neq 1$, and how are the spaces $\mathcal{H}_\chi$'s related for inequivalent $\chi$'s?}

For flat gauge fields of finite groups, this question has been studied in the context of discrete gauging \cite{Cheng:2022sgb,Seifnashri:2023dpa,Seiberg:2023cdc,Choi:2024rjm,Seifnashri:2026ema}. The answer is that the modification of Gauss law at one site\footnote{Here we consider $1$-form gauge field. The stories are the same for higher form gauge fields.} corresponds to inserting a Wilson line along the time direction at this site. The Wilson line is topological since the gauge field is flat. 

The answer can be seen as follows. For a flat $1$-form gauge field, we only need a pair conjugate variables on each links. If we denote the local matter charge as $Q_s$ and the conjugate variable to $A_\ell$ as $E_\ell$, than the Gauss law operator is
\begin{align}\label{eq:simple case Gauss law}
    G(\zeta)=\exp\left(i\sum_s \zeta_s (Q+\delta E)_s\right),
\end{align}
for $\zeta_s$ being an element in the finite gauge group. See Appendix \ref{sec:convensions} for the definition of $\delta$.  The Wilson lines generate a $(D-2)$-form symmetry \footnote{i.e. the quantum symmetry in the context of discrete gauging.} \cite{Kapustin:2014gua,Gaiotto:2014kfa,Cordova:2022ruw,Schafer-Nameki:2023jdn,Costa:2024wks,Kaidi:2026urc}, where $D$ is the spacetime dimension. Inserting a Wilson line $\exp(ik\int A)$ at site $s_*$ is done by conjugating the system (i.e. Gauss law and Hamiltonian) by the unitary 
\begin{align}
    \exp\left(ik\sum_{\ell\in L} A_\ell\right),
\end{align}
where $L$ is a half-line ending on $s_*$. This shift $G$ by a phase only depends on $\zeta_{s_*}$,
\begin{align}
    G(\zeta)\to e^{ik\zeta_{s_*}}G(\zeta).
\end{align}
It can be checked that starting with $\mathcal{H}_1$  defined by \eqref{eq:intro the Gauss law}, one can obtain $\mathcal{H}_\chi$ for any nontrivial $\chi$ by inserting appropriate Wilson lines at each site. Equivalently, the physical Hilbert space for a modified Gauss law is the defect Hilbert space, twisted by some Wilson lines.

For a dynamical $\mathrm{U}(1)$ gauge field \cite{Cheng:2022sgb,Seifnashri:2026ema,Lu:2026itw,Lu:2026jnq,Chatterjee:2024gje,Witten:2003ya}, because of the non-flat field strength, we can have a monopole configuration. To realize the monopole operator and the monopole number operator on the lattice, we need to introduce the Villain field $N$, which is a $2$-from $\mathbb{Z}$ gauge field. One way to think about this is that we start with a $\mathbb{R}$ gauge field and than flat gauge the $\mathbb{Z}^{(1)}$ symmetry \cite{Gorantla:2021svj,Villain:1974ir,Sulejmanpasic:2019ytl,Cordova:2019jnf,Cordova:2019uob,Yoneda:2022qpj,Fazza:2022fss,Chen:2024ddr,Xu:2024hyo,Jacobson:2024hov,Cherman:2024exo}. Or we just view $N$ as one of the differential cohomology defining data for a $\mathrm{U(1)}$ gauge field in its own right \cite{cheeger1985differential,hopkins2005quadratic,Freed:2006yc,Aharony:2016kai,Hsieh:2020jpj}. See \cite{Hsieh:2020jpj} for a pedagogical introduction, and the $N$ and $A$ in our notation is almost the discrete version of theirs.

Now, with more gauge parameters and more subtle gauge transformations, the Gauss law operator takes a more complicated form than \eqref{eq:simple case Gauss law}. Without the previous $(D-2)$-from symmetry since the Wilson lines are \emph{not} topological, one might expect the $\mathrm{U(1)}^{(D-3)}$ magnetic symmetry defect will save the day. But as we show in Section \ref{sec:Moduli space and dynamical defects}, these topological defect insertions do \emph{not} generate all the modified Gauss laws, i.e. nontrivial $\chi$'s.

\paragraph{Main results}

In this work we focus on $\mathrm{U}(1)$ pure gauge theory in
$2+1$d.  We show that the consistent modified Gauss laws are classified by real lattice 1-cochain $\lambda\in C^1(\Lambda,\mathbb{R})$ modulo closed
$2\pi\mathbb{Z}$-valued 1-cochains and a $2\pi$ periodic 2-cochian $\theta\in C^2(\Lambda,\mathbb{R}/2\pi\mathbb{Z})$.  More explicitly, the modified  local Gauss laws are
\begin{equation}
    (\delta E)_s=\frac{(\delta\lambda)_s}{2\pi},
    \qquad
    e^{i(2\pi E+\delta B)_\ell}=e^{i\lambda_\ell}, \quad e^{2\pi iN_p}=e^{i\theta_p}
\end{equation}

We then identify the physical meaning of these sectors. A non-zero $\lambda$ corresponds to inserting \emph{open} magnetic surface defects
\begin{equation}
    \exp\left(i\lambda\int_\Sigma\frac{F}{2\pi}\right),
    \qquad \partial\Sigma\neq0.
\end{equation}
The boundary of such an open surface is a monodromy defect for monopole
operators, also referred to as the twist defect.  On the lattice, the monodromy at a site $s$ is
$e^{i(\delta\lambda)_s}$. And a non-zero $\theta$ corresponds to electric $\mathrm{U(1)}_e^{(1)}$ topological defects, or a vortex defect when the electric $\mathrm{U(1)}_e^{(1)}$ symmetry is explicitly broken by some matter fields.

Then, we construct the translation operators acting on these defect
Hilbert spaces $\mathcal{H}_\chi$. These translations obey
\begin{equation}
    T_1T_2T_1^{-1}T_2^{-1}
    =
    \exp\left[
    -i\sum_s(\delta\lambda)_s
    \left(N+\frac{dA}{2\pi}\right)_{s+\frac{\hat x}{2}+\frac{\hat y}{2}}
    \right],
\end{equation}
which is a magnetic-translation algebra.  The noncommutativity of
translations is therefore the manifestation of the monodromy defects.

Finally we demonstrate the exact duality to compact boson as a Hamiltonian lattice model, and compare the modified Gauss law on both sides.

The paper is organized as follows. In Section \ref{sec:2+1d dynamical U(1) gauge theory}, we explain the modified Gauss law in dynamical $\mathrm{U(1)}$ gauge theory, and work out the moduli space of physical Hilbert spaces. Section \ref{sec:Moduli space and dynamical defects} shows that a non-trivially modified Gauss law describes the Hilbert space with dynamical monodromy defect insertions. In Section \ref{sec:magnetic translation}, the translation operation on each Hilbert space is constructed, with the symmetry algebra agreeing with the monodromy. Section \ref{sec:duality} is about the exact lattice model duality and some discussions about matter are made in Section \ref{sec:Discussions and future directions}.

Our conventions and notations are collected in Appendix \ref{sec:convensions}, a proof is given in Appendix \ref{sec:proof of lambda for chi}, and the comparison with the Hamiltonian lattice gauge theory of a finite group is made in Appendix \ref{app:finite}.

\section{2+1d dynamical $\mathrm{U(1)}$ gauge theory}\label{sec:2+1d dynamical U(1) gauge theory}

\subsection{Modified Gauss law}

We explain the modification of Gauss law in a Hamiltonian lattice model, with only $\mathrm{U(1)}$ gauge field, and without matter degrees of freedom.

The spatial lattice is a 2d square lattice with periodic boundary conditions along both $x$ and $y$ directions. We denote $\Lambda$ for the lattice, $s$ for sites, $\ell$ for links and $p$ for plaquettes. 

We recall the Hilbert space for dynamical $\mathrm{U(1)}$ gauge fields on the lattice as follows. The idea follows from Villain \cite{Villain:1974ir} where one represents a $\mathrm{U(1)}$ gauge field by an $\mathbb{R}$ gauge field and a $\bZ$ gauge field, see also recent discussions \cite{Sulejmanpasic:2019ytl,Fazza:2022fss,Lu:2026jnq,Chen:2024ddr,Cordova:2019jnf,Cordova:2019uob}. On each link we place a pair of $\mathbb{R}$-valued conjugate variables,
\begin{align}\label{eq:abelian gauge theory commutator 1}
    [A_\ell, E_{\ell'}]=i\delta_{\ell\ell'},
\end{align}
and the local Hilbert space on each link forms a representation of this algebra $\mathcal{H}_\ell \simeq L^2(\mathbb{R})$. On each plaquette we place a pair of $\mathbb{R}$-valued conjugate variables,
\begin{align}\label{eq:abelian gauge theory commutator 2}
    [B_p, N_{p'}]=i\delta_{pp'},
\end{align}
and the local Hilbert space on each plaquette forms a representation of this algebra $\mathcal{H}_p \simeq L^2(\mathbb{R})$. We refer to the tensor product Hilbert space $\mathcal{H}=\otimes_p \mathcal{H}_p\otimes_\ell \mathcal{H}_\ell$ as the extended Hilbert space. 

The tensor product Hilbert space is subjected to two types of Gauss laws. The first Gauss law enforces the eigenvalue of $N_p$ to be integers by demanding 
\begin{equation}\label{eq:gausslaw1}
    \exp\left( 2\pi i \widetilde{n}_p N_p\right)
\end{equation}
to be 1 for every plaquette, for an arbitrary 
\begin{equation}
    \widetilde{n}\in C^2(\Lambda,\mathbb{Z})\,.
\end{equation}
Later, we allow more general Gauss laws by allowing \eqref{eq:gausslaw1} to be an arbitrary constant phase, not necessarily 1. This is why we start by setting $N_p$ to be $\mathbb{R}$ valued instead of $\bZ$ valued. The Gauss law of $N_p$ also generates a gauge transformation of its conjugate variable $B_p$ 
\begin{align}\label{eq:U(1) gauge transformation 2}
    B_p\to B_p+2\pi \widetilde{n}_p\,. 
\end{align}
This Gauss law does not break the tensor product structure of $\mathcal{H}$. 

To see the second Gauss law, it is more natural to start with the gauge transformation. The $\mathrm{U(1)}$ gauge field is represented by a pair $(N_p, A_\ell)$, with the gauge transformation 
\begin{align}\label{eq:U(1) gauge transformation 1}
\begin{split}
    &A_\ell\to A_\ell+2\pi n_\ell+(d\alpha)_\ell,\\
    &N_p \to N_p-(dn)_p,
\end{split}
\end{align}
where 
\begin{align}
    n\in C^1(\Lambda, \bZ),\quad \alpha\in C^0(\Lambda, \mathbb{R}).
\end{align}
More concretely, $n_\ell$'s are $\mathbb{Z}$-valued numbers on links and $\alpha_s$'s are $\mathbb{R}$-valued numbers on sites. See Appendix \ref{sec:convensions} for the conventions for (co)chains, operator $d$ and its adjoint $\delta$. The gauge transformation \eqref{eq:U(1) gauge transformation 1} is generated by two Gauss law operators
\begin{equation}\label{eq:gausslaw2}
    \exp\left(i\sum_\ell  n_\ell (2\pi E_\ell + (\delta B)_{\ell})\right), \qquad \exp\left( i \sum_s\alpha_s (\delta E)_s \right)
\end{equation}
The first one generates the $\bZ$ redundancy of $A_\ell$, which effectively makes it compact. The second one is the standard Gauss law for $U(1)$.\footnote{In the continuum we typically refer to the transformation associated with the gauge group as the gauge transformation, which is \eqref{eq:U(1) gauge transformation 1} in this case. And we call \eqref{eq:U(1) gauge transformation 2} the identification in the definition of a periodic scalar, rather than treating it as a Gauss law. However, on the lattice both these are realized as constraints, or the Gauss laws, so they should not be view as different kinds. In fact, in Section \ref{sec:duality}, we see this two roles get exchanged under the duality. \label{fn:gauss laws and constraints}} This Gauss law breaks the tensor product structure of $\mathcal{H}$.

One can package the Gauss laws \eqref{eq:gausslaw1} and \eqref{eq:gausslaw2} into a single one 
\footnote{The term exponents can not be placed in a single exponent because $(\delta B)_\ell$ does not commute with $N_p$ when $\ell \in \partial p$. But the two exponents commute.}
\begin{align}\label{eq:U(1) Gauss law operator}
    G(n,\alpha,\widetilde{n})=\exp\left(i\sum_\ell n_\ell\left(2\pi E_\ell+(\delta B)_\ell\right)+i\sum_s \alpha_s(\delta E)_s\right) \exp\left( 2\pi i\sum_p \widetilde{n}_pN_p\right).
\end{align}

The operator $\delta: C^k(\Lambda)\to C^{k-1}(\Lambda)$ can be thought as the lattice analogy for divergence and is defined in Appendix \ref{sec:convensions}. Note that instead of writing the Gauss law operator as a local operator, say at one site (for one non-zero $\alpha_s$), we defined it as the operator implementing any gauge transformation parameterized by $(n,\alpha, \widetilde{n})\in C^1(\Lambda, \mathbb{Z})\times C^0(\Lambda, \mathbb{R})\times C^2(\Lambda,\mathbb{Z})$. This significantly simplifies the discussions below.

The physical Hilbert space is a subspace $\mathcal{H}_\chi$ of the tensor product Hilbert space $\mathcal{H}$. More concretely, $\mathcal{H}_\chi$ is an eigenspace of the Gauss law $G(n,\alpha,\widetilde{n})$ with eigenvalue $\chi(n,\alpha,\widetilde{n})\in \mathrm{U(1)}$, i.e. 
\begin{align}\label{eq:U(1) physical Hilbert space}
    \mathcal{H}_\chi=\{\ket{\psi}\in\mathcal{H}:G\ket{\psi}=\chi\ket{\psi}\}\subset\mathcal{H}\,.
\end{align}
Traditionally \cite{Cheng:2022sgb,Lu:2026jnq}, $\chi$ is taken to be 1 hence $G(n,\alpha,\widetilde{n})=1$, reducing to 
\begin{align}\label{eq:trandidtional Gauss law}
\begin{split}
    e^{i(2\pi E_\ell+(\delta B)_\ell)}=1,\qquad (\delta E)_s=0, \qquad e^{2\pi iN_p}=1\,.
\end{split}
\end{align}
\eqref{eq:trandidtional Gauss law} is usually referred to as \emph{the} Gauss law for $\mathrm{U(1)}$ gauge theory on the lattice. But generally, $\chi$ can be a $\mathrm{U(1)}$-valued function
\begin{align}
    \chi:\quad C^1(\Lambda, \mathbb{Z})\times C^0(\Lambda, \mathbb{R})\times C^2(\Lambda,\mathbb{Z})\to \mathrm{U(1)}.
\end{align}
We refer to $G=\chi$ for nontrivial $\chi$ as a \emph{modified Gauss law}.

We make some remarks on $\chi$. 
\begin{enumerate}
    \item Firstly, different terms in \eqref{eq:U(1) Gauss law operator} commute. In particular, the $\exp(i n_\ell (\delta B)_\ell)$ from the first term commutes with $\exp(2\pi i \widetilde{n}_p N_p)$ from the second term.  It follows that $\chi$ is a group homomorphism between abelian groups, i.e. $\chi(n_1+n_2, \alpha,\widetilde{n})=\chi(n_1,\alpha,\widetilde{n})\chi(n_2,\alpha,\widetilde{n})$, etc. 
    \item Secondly, $\chi$ satisfies some consistency conditions, such as $\chi(dk,-2\pi k,0)=G(dk,-2\pi k,0)=1$ by the definition of $G$ \eqref{eq:U(1) Gauss law operator}. We will work out the space of consistent $\chi$'s \eqref{eq:moduli space for U(1) gauge theory}, that is, the moduli space physical Hilbert space $\mathcal{H}_\chi$. Each point in the moduli space fixes a physical Hilbert space on the lattice. 
    \item Thirdly, we only specify the gauge field degrees of freedom. The summation over site terms in $G$ \eqref{eq:U(1) Gauss law operator} should include matter local charges if we intend to include the matter. Then the moduli space of $\chi$ will be different. Here we focus on pure gauge theory. Including matter is briefly discussed in Section \ref{sec:Discussions and future directions}.
\end{enumerate}

\subsection{Moduli space of physical Hilbert spaces}\label{sec:Moduli space of physical Hilbert spaces}
We work out a way to characterize different $\chi$'s, hence obtaining the moduli space for physical Hilbert spaces $\mathcal{H}_\chi$.

Being a group homomorphism, $\chi(n,\alpha,\widetilde{n})=\chi(n,0,0)\chi(0,\alpha,0)\chi(0,0,\widetilde{n})$, $\chi(0,0,\widetilde{n})$ can be easily fixed, because $N_p$ only appears in the third factor of $G$. But the variable $E_\ell$ is involved in the first two factors of $G$ simultaneously, so $\chi(n,\alpha,0)$ needs to be treated with more care.

We start by looking at $\chi(0,0,-):C^2(\Lambda,\mathbb{Z})\to \mathrm{U(1)}$, since any one dimensional representation of $\mathbb{Z}$ is labeled by a $\theta\in \mathbb{R}/2\pi\mathbb{Z}$, i.e. $\rho_\theta(n)=e^{i\theta n}$, there exists a unique set of $\{\theta_p\in\mathbb{R}/2\pi\mathbb{Z}\}$, i.e. $\theta\in C^2(\Lambda, \mathbb{R}/2\pi\mathbb{Z})$, such that 
\begin{align}
    \chi(0,0,\widetilde{n})=\exp\left(i\sum_p \widetilde{n}_p \theta_p\right).
\end{align}

Then, for $\chi(0,-,0):C^0(\Lambda,\mathbb{R})\to \mathrm{U(1)}$, since any one dimensional representation of $\mathbb{R}$ is labeled by a real number $q\in\mathbb{R}$, i.e. $\rho_q(x)=e^{iqx}$, there exists a unique set of $\{q_s\in\mathbb{R}\}$, i.e, $q\in C^0(\Lambda,\mathbb{R})$, such that 
\begin{align}
    \chi(0,\alpha,0)=\exp\left(i\sum_s \alpha_sq_s\right).
\end{align}

Lastly, for $\chi(-,0,0): C^1(\Lambda,\mathbb{Z})\to \mathrm{U(1)}$, there exists a unique set of $\{\varphi_\ell\in\mathbb{R}/2\pi\mathbb{Z}\}$, i.e. $\varphi\in C^1(\Lambda, \mathbb{R}/2\pi\mathbb{Z})$, such that 
\begin{align}
    \chi(n,0,0)=\exp\left(i\sum_\ell n_\ell\varphi_\ell\right).
\end{align}

We claim that there exist $\lambda\in C^1(\Lambda, \mathbb{R})$, such that
\begin{align}\label{eq:existing lambda for chi}
    q_s=\frac{(\delta \lambda)_s}{2\pi}\in\mathbb{R},\quad \text{and} \quad \varphi_\ell=\lambda_\ell \mod 2\pi.
\end{align}
This is proven in Appendix \ref{sec:proof of lambda for chi}, where $\lambda$ is constructed, using some consistency conditions in $\chi$ mentioned in the above remarks. This is not an obvious fact which can easily be seen by modifying the right-hand side of the traditional local Gauss law \eqref{eq:trandidtional Gauss law}. Thus, the function $\chi$ can be written as
\begin{align}\label{eq:chi from lambda}
    \chi(n,\alpha,\widetilde{n})=\exp\left(i\sum_\ell n_\ell\lambda_\ell+i\sum_s \alpha_s \frac{(\delta\lambda)_s}{2\pi}+i\sum_p \widetilde{n}_p\theta_p\right).
\end{align}
Thus, on the physical Hilbert space
\begin{align}\label{eq:local Gauss law}
    (\delta E)_s=\frac{(\delta \lambda)_s}{2\pi},\quad e^{i(2\pi E+\delta B)_\ell}=e^{i\lambda_\ell},\quad e^{2\pi iN_p}=e^{i\theta_p}
\end{align}

Note that $\lambda$ and $\lambda'$ determine the same $\chi$ by \eqref{eq:chi from lambda} if and only if $\lambda'-\lambda\in C^1(\Lambda,2\pi \mathbb{Z})$ and $\delta\lambda=\delta\lambda'$.  Thus, the space of inequivalent $\lambda$ and $\theta$, is 
\begin{align}\label{eq:moduli space for U(1) gauge theory}
    \mathcal{M}=\frac{C^1(\Lambda,\mathbb{R})}{C^1(\Lambda,2\pi \mathbb{Z})\cap\ker \delta}\times C^2(\Lambda,\mathbb{R}/2\pi\mathbb{Z}).
\end{align}
Any representative element $(\lambda,\theta)\in\mathcal{M}$ fixes a $\chi$ by \eqref{eq:chi from lambda}\footnote{We identify $\mathcal{M}$ with the space of $\chi$'s, and just write $\chi\in\mathcal{M}$ hoping without confusion.}, hence by \eqref{eq:U(1) physical Hilbert space} fixing a physical Hilbert space $\mathcal{H}_\chi$ of pure $\mathrm{U(1)}$ gauge theory. We say that $\mathcal{M}$ is the \emph{moduli space of physical Hilbert spaces}.

\subsection{Different presentations}
Our discussion suffices for any Hamiltonian lattice model with constraints, not necessarily those coming from the obvious analog of the gauge theory in the continuum, one example will be given in Section \ref{sec:duality}. So what we mean by the modified Gauss laws are actually the modified constraints on physical Hilbert spaces. We use \emph{the Gauss laws} and \emph{the constraints} interchangeably.

The lattice model is defined by the pair $(G=\chi, H)$, where $H$ is the Hamiltonian. It is known \cite{Seifnashri:2026ema,Lu:2026itw} that we can perform some unitary transformation to the system, shifting between different presentations of the theory, for example from the modified Gauss laws to the modified Hamiltonian $(G'=UGU^{-1}=1, H'=U H U^{-1})$. 

Note that the unitary $U$ is the operator on the total Hilbert space $\mathcal{H}$, not just the unitary operators on physical space $\mathcal{H}_\chi$. Although the unitary transformation does not change the spectrum of $H$ in $\mathcal{H}$, the spectrum of $H|_{\mathcal{H}_\chi}$ is changed since the physical Hilbert space is changed. We will see such examples below. This is consistent with the fact that modification of the constraints changes the energy spectrum. 

Consider the following standard Maxwell theory Hamiltonian on the lattice,
\begin{align}\label{eq:unmodified Hamiltonian}
    H=\sum_p\frac{1}{2e^2}\left(2\pi N+dA\right)_p^2+\sum_\ell \frac{e^2}{2}E_\ell^2,
\end{align}
with the \emph{modified} constraints \eqref{eq:local Gauss law}, where the two terms are the energy for magnetic and electric fields, respectively. We can perform the unitary to do the shift,
\begin{align}\label{eq:unitary to un-mopdified Gauss law}
    U:\quad E \to E+\frac{\lambda}{2\pi},\quad N\to N+\frac{\theta}{2\pi},
\end{align}
which is realized by $U=\exp\left(-\tfrac{i}{2\pi}\langle\lambda,A\rangle-\tfrac{i}{2\pi}\langle\theta,B\rangle\right)$ (see \eqref{eq:inner product of cochain} for the shorthand notation of inner product). In this new presentation, the Gauss laws are unmodified as \eqref{eq:trandidtional Gauss law}, but the Hamiltonian gets modified
\begin{align}\label{eq:modified gauge theory Hamiltonian}
    H'=\sum_p\frac{1}{2e^2}\left(2\pi N+dA+\theta\right)_p^2+\sum_\ell \frac{e^2}{2}\left(E+\frac{\lambda}{2\pi}\right)_\ell^2.
\end{align}
See the summary in Table \ref{tab:2 presentations}. Of course, we can introduce some intermediate presentation, by conjugating the unitary which only shift one of $N$ and $E$, where both the Gauss law and Hamiltonian get modified. We will use different presentations interchangeably, since they defined the same theory.

\begin{table}
    \centering
    \begin{tabular}{c|c|c}
       Presentations & 
       \begin{tabular}{c}
            Modified Gauss law,  \\
            Unmodified Hamiltonian 
       \end{tabular}
        & \begin{tabular}{c}
            Unmodified Gauss law,  \\
            Modified Hamiltonian 
       \end{tabular} \\
       \hline
       Gauss laws&
       $\begin{aligned}
           \delta E &= \frac{\delta \lambda}{2\pi}\\
           e^{i(2\pi E+\delta B)} &= e^{i\lambda}\\
           e^{2\pi iN} &= e^{i\theta}
       \end{aligned}$
       &
       $\begin{aligned}
           \delta E &= 0\\
           e^{i(2\pi E+\delta B)} &= 1\\
           e^{2\pi iN} &= 1
       \end{aligned}$
       \\
       \hline
       Hamiltonian &
       $\begin{aligned}
           H=\sum_p&\frac{1}{2e^2}\left(2\pi N+dA\right)_p^2\\
           &+\sum_\ell \frac{e^2}{2}E_\ell^2
       \end{aligned}$
       &
       $\begin{aligned}
           H'=\sum_p&\frac{1}{2e^2}\left(2\pi N+dA+\theta\right)_p^2\\
           &+\sum_\ell \frac{e^2}{2}\left(E+\frac{\lambda}{2\pi}\right)_\ell^2
       \end{aligned}$
    \end{tabular}
    \caption{Two presentations of the modified Maxwell theory on the lattice. The transformation between the two presentations is achieved by the unitary \eqref{eq:unitary to un-mopdified Gauss law}.}
    \label{tab:2 presentations}
\end{table}

Note that since we have classified all possible modified Gauss laws in Section \ref{sec:Moduli space of physical Hilbert spaces}, which is independent of the Hamiltonian, the unitary transformation \eqref{eq:unitary to un-mopdified Gauss law} to another presentation works for any Hamiltonian that commutes with the Gauss law.

\section{Moduli space and defects}\label{sec:Moduli space and dynamical defects}

To understand the meaning of nontrivial $\chi$'s, we need to take a digression and discuss inserting (dynamical) defect in the Hamiltonian lattice model. Here defects are to be contrasted with operators, where the former stretch in the time direction.

We show that for a nontrivial $\chi$ \eqref{eq:chi from lambda}, determined by $(\lambda,\theta)\in\mathcal{M}$, the non-zero $\theta$ corresponds to the insertion of topological $\mathrm{U(1)}_e^{(1)}$ line defects, and the non-zero $\lambda$ corresponds to the insertion of \emph{non-topological} monodromy (or twist) defects for $\mathrm{U(1)}_m^{(0)}$ open surface defect. Hence the physical Hilbert space $\mathcal{H}_\chi$ is the corresponding defect Hilbert space for 2+1d $\mathrm{U(1)}$ pure gauge theory.

\subsection{Symmetries}\label{sec:symmetries}
We work in the first presentation where the Gauss laws are modified \eqref{eq:local Gauss law} but the Hamiltonian is unmodified \eqref{eq:unmodified Hamiltonian}. The theory has a magnetic 0-from $\mathrm{U(1)}_m^{(0)}$, with the charge operator
\begin{align}\label{eq:magnetic 0-form charge}
    Q^{(0)}=\sum_p \left(N-\frac{\theta}{2\pi}\right)_p.
\end{align}
Eigenvalues of the local charge $N-\theta/2\pi$ are integers because of the third equation in the Gauss laws \eqref{eq:local Gauss law}. Though $N_p$ is not gauge invariant, $Q^{(0)}$ is since $\sum_pN_p=\sum_p(N+dA/2\pi)_p$, where each summand is gauge invariant. The charged point operator is the monopole $e^{iB_p}$, which is not involved in the Hamiltonian, implying $[Q^{(0)},H]=0$.

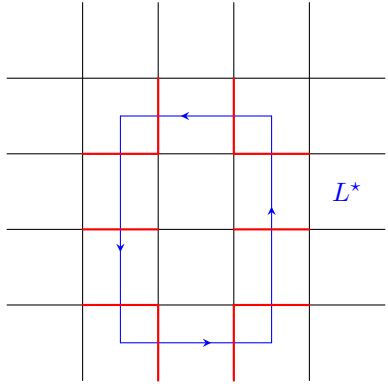
\begin{figure}
\centering
\begin{tikzpicture}[
    scale=1,
    line join=round,
    line cap=round,
    every node/.style={font=\footnotesize}
]
    \draw (0,1)--(5,1);
    \draw (0,2)--(5,2);
    \draw (0,3)--(5,3);
    \draw (0,4)--(5,4);
    \draw (1,0)--(1,5);
    \draw (2,0)--(2,5);
    \draw (3,0)--(3,5);
    \draw (4,0)--(4,5);
    \draw[blue, decoration = {markings, mark=at position 0.6 with {\arrow[scale=1.0]{stealth}}}, postaction=decorate] (1.5,3.5) -- (1.5,0.5);
    \draw[blue, decoration = {markings, mark=at position 0.6 with {\arrow[scale=1.0]{stealth}}}, postaction=decorate] (1.5,0.5) -- (3.5,0.5);
    \draw[blue, decoration = {markings, mark=at position 0.6 with {\arrow[scale=1.0]{stealth}}}, postaction=decorate] (3.5,0.5) -- (3.5,3.5);
    \draw[blue, decoration = {markings, mark=at position 0.6 with {\arrow[scale=1.0]{stealth}}}, postaction=decorate] (3.5,3.5) -- (1.5,3.5);
    \node[blue] at (4.5,2.5) {$L^\star$};
    \draw[red,thick] (3,3) -- (3,4);
    \draw[red,thick] (2,3) -- (2,4);
    \draw[red,thick] (1,3) -- (2,3);
    \draw[red,thick] (1,2) -- (2,2);
    \draw[red,thick] (1,1) -- (2,1);
    \draw[red,thick] (3,0) -- (3,1);
    \draw[red,thick] (2,0) -- (2,1);
    \draw[red,thick] (3,3) -- (4,3);
    \draw[red,thick] (3,2) -- (4,2);
    \draw[red,thick] (3,1) -- (4,1);
\end{tikzpicture}
    \caption{Symmetry defect for electric 1-from $\mathrm{U(1)}_e^{(1)}$ symmetry.}
    \label{fig:symmetry defect for electric 1-from U(1) symmetry}
\end{figure}

The theory also has an electric 1-from $\mathrm{U(1)}_e^{(1)}$ symmetry, with the charge operator
\begin{align}\label{eq:electic 1-form charge}
    Q^{(1)}=\sum_{\ell^\star\in L^\star} \left(E+\frac{\delta B}{2\pi}-\frac{\lambda}{2\pi}\right)_\ell,
\end{align}
where $L^\star$ is the closed loop in the dual lattice, as in Figure \ref{fig:symmetry defect for electric 1-from U(1) symmetry}. Eigenvalues of the local charge are integers because of the second equation in the Gauss laws \eqref{eq:local Gauss law}. $Q^{(1)}$ is topological because of $[Q^{(1)},H]=0$ (in temporal direction) and $\delta(E-\delta B/2\pi-\lambda/2\pi)=0$ (in spatial directions), which is also a part of the Gauss laws. The charged line operator is the Wilson line $\exp(i\sum_{\ell\in\gamma}A_\ell)$.

Even though $[Q^{(0)},Q^{(1)}]=0$, the local charges do not commute, as a manifestation the mixed anomaly similar to \cite{Seifnashri:2026ema}.

The continuous counterparts of the charge operators are 
\begin{align}
    \mathcal{Q}^{(0)}=\int_\Sigma\frac{f}{2\pi}, \quad \text{and}\quad \mathcal{Q}^{(1)}=\int_L\frac{i}{e^2}\star f.
\end{align}
Strictly speaking, we have not derived the interpretation of the modified Gauss law, so the operator in the continuum is valid away from the spacetime region where modifications are made.

\subsection{Topological defects}\label{sec:topological defect}

If the defect is topological, inserting such symmetry defects has been studied extensively in \cite{Cheng:2022sgb,Seifnashri:2023dpa,Seiberg:2023cdc,Choi:2024rjm,Chatterjee:2024gje,Seifnashri:2025fgd,Seiberg:2024gek,Seifnashri:2025vhf,Li:2023knf,Cao:2023doz,Cao:2022lig}. 

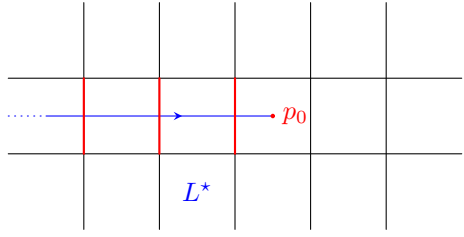
\begin{figure}
\centering
\begin{tikzpicture}[
    scale=1,
    line join=round,
    line cap=round,
    every node/.style={font=\footnotesize}
]
    \draw (0,1)--(6,1);
    \draw (0,2)--(6,2);
    \draw (1,0)--(1,3);
    \draw (2,0)--(2,3);
    \draw (3,0)--(3,3);
    \draw (4,0)--(4,3);
    \draw (5,0)--(5,3);
    \draw[blue, decoration = {markings, mark=at position 0.6 with {\arrow[scale=1.0]{stealth}}}, postaction=decorate] (0.5,1.5) -- (3.5,1.5);
    \draw[blue, dotted] (0,1.5) -- (0.5,1.5);
    \node[circle, fill=red, inner sep=0.6pt] at (3.5,1.5) {};
    \node[blue] at (2.5,0.5) {$L^\star$};
    \node[red] at (3.8,1.5) {$p_0$};
    \draw[red,thick] (1,1) -- (1,2);
    \draw[red,thick] (2,1) -- (2,2);
    \draw[red,thick] (3,1) -- (3,2);
\end{tikzpicture}
    \caption{Half-line operator $U^{(1)}(\alpha,p_0)$ to insert a $\mathrm{U(1)}_e^{(1)}$ defect at plaquette $p_0$.}
    \label{fig:Half-line operator for electric 1-from U(1) symmetry}
\end{figure}

We start by inserting a $\mathrm{U(1)}_e^{(1)}$ defect at plaquette $p_0$. This is achieved by conjugating the system (including both the Gauss laws and the Hamiltonian) by a truncated $U(1)_e^{(1)}$ symmetry operator $e^{i \alpha Q^{(1)}}$ where $Q^{(1)}$ is \eqref{eq:electic 1-form charge} but defined on truncated line $L^\star$ with a boundary at $\partial L^\star=p_0$ as in Figure \ref{fig:Half-line operator for electric 1-from U(1) symmetry},
\begin{align}
\begin{split}
    U^{(1)}(\alpha,p_0)=\exp\left[i\alpha \sum_{L^\star}\left(E-\frac{\delta B}{2\pi}-\frac{\lambda}{2\pi}\right)_\ell\right]\\
    =\exp\left[i\alpha \sum_{L^\star}\left(E-\frac{\lambda}{2\pi}\right)_\ell\right]\exp\left(i\frac{\alpha}{2\pi}B_{p_0}\right).
\end{split}
\end{align}
Conjugation by $U^{(1)}(\alpha,p_0)$ shifts $A_\ell\to A_\ell + \alpha$ for those $\ell^\star\in L^\star$ (red links in Figure \ref{fig:Half-line operator for electric 1-from U(1) symmetry}), and shifts $N_{p_0}\to N_{p_0}-\alpha/2\pi$. Thus, the Hamiltonian $H$ \eqref{eq:unmodified Hamiltonian} and the first two Gauss law operators \eqref{eq:local Gauss law} are invariant, while the last Gauss law operator gets a phase $e^{-i\alpha}$ at $p_0$. This means that inserting a $\mathrm{U(1)}_e^{(1)}$ defect at plaquette $p_0$ changes the constraint from $G(n,\alpha,\widetilde{n})=\chi(n,\alpha,\widetilde{n})$ to $G(n,\alpha,\widetilde{n})=\chi(n,\alpha,\widetilde{n}) e^{i\widetilde{n}_{p_0}\alpha}$, or equivalently
\begin{align}\label{eq:shifting theta}
    (\lambda,\theta)\to(\lambda,\theta+\alpha1_{p_0}),
\end{align}
where $1_{p_0}$ is a 2-cochain which is $1$ on plaquette $p_0$ and $0$ elsewhere.

Note that this is a transitive action on the $C^2(\Lambda,\mathbb{R}/2\pi\mathbb{Z})$ factor of the moduli space $\mathcal{M}$ \eqref{eq:moduli space for U(1) gauge theory}. Thus, the non-zero $\theta$ in the modified Gauss laws corresponds to the insertion of topological $\mathrm{U(1)}_e^{(1)}$ line defects. If we shift to the second presentation in Table \ref{tab:2 presentations}, this agrees with (4.27) in \cite{Komargodski:2025jbu}.

Now consider a $\mathrm{U(1)}_m^{(0)}$ topological defect supported on a closed surface $\Sigma$ with parameter $\alpha\in \mathbb{R}/2\pi\mathbb{Z}$. In the spacetime picture, the surface $\Sigma$ intersects the equal time slice on a closed cycle $L_0$, as in Figure \ref{fig:spacetime U(1) defect}. On the lattice, inserting this defect on $L_0$ is achieved by conjugating the system by a truncated $\mathrm{U(1)}_m^{(0)}$ defect $e^{i \alpha Q^{(0)}}$ where $Q^{(0)}$ is \eqref{eq:magnetic 0-form charge} but defined on truncated surface $D$ with a boundary at $\partial D= L_0$ as in Figure \ref{fig:inserting topological defect on the lattice},
\begin{align}\label{eq:U0twist}
    U^{(0)}(\alpha,L_0)=\exp\left(i\alpha\sum_{p\in D}\left(N-\frac{\theta}{2\pi}\right)_p\right).
\end{align}
Conjugating by $U^{(0)}(\alpha,L_0)$ shifts $B_p\to B_p + \alpha$ for those $p\in D$. Again, the Hamiltonian $H$ \eqref{eq:unmodified Hamiltonian}, the first and the last Gauss law operators \eqref{eq:local Gauss law} are invariant, while the second Gauss law operator gets a phase $e^{i\alpha}$ at those $\ell\in L_0$. This means that inserting a $\mathrm{U(1)}_e^{(1)}$ defect at $L_0$ changes the modified Gauss laws
\begin{align}\label{eq:topolocal defect shifts lambda}
    (\lambda,\theta)\to(\lambda-\alpha1_{L_0},\theta).
\end{align}
Here $1_{L}$ is an integer 1-cochain which is $+1$ on the links of $L$ with positive orientation, $-1$ on links taken the negative orientation, and $0$ elsewhere.

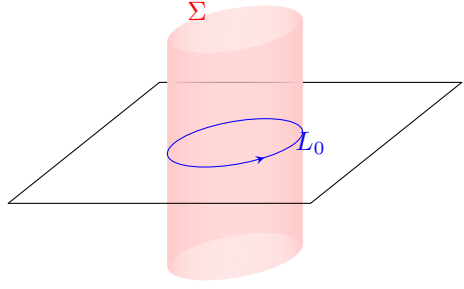
\begin{figure}
\centering
\begin{tikzpicture}[
    x={(-0.5cm,-0.4cm)},  
    y={(1cm,0cm)},         
    z={(0cm,1cm)},         
    scale=1,
    line join=round,
    line cap=round,
    every node/.style={font=\footnotesize}
]
\draw (2,-2,0) -- (-2,-2,0) -- (-2,2,0) -- (2,2,0);
\foreach \zz in {-1.5,-1.49,...,1.5} {
    \begin{scope}[canvas is xy plane at z=\zz]
        \fill[fill=red!20,opacity=0.5, even odd rule]
            (0,0) circle (0.8)
            (0,0) circle (0.77);
    \end{scope}
}
\draw (2,-2,0) -- (2,2,0);
\draw[blue,decoration = {markings, mark=at position 0.2 with {\arrow[scale=1.0]{stealth}}}, postaction=decorate] (0,0,0) circle (0.8);
\node[blue] at (0,1,0) {$L_0$};
\node[red] at (1,0,2.15) {$\Sigma$};
\end{tikzpicture}
    \caption{The spacetime picture for the $\mathrm{U(1)}^{(0)}_m$ surface defect intersects with the equal time slice on a closed cycle $L$.}
    \label{fig:spacetime U(1) defect}
\end{figure}

\begin{figure}
\centering
\begin{tikzpicture}[
    scale=1,
    line join=round,
    line cap=round,
    every node/.style={font=\footnotesize}
]
\draw (0,1)--(5,1);
    \draw (0,2)--(5,2);
    \draw (0,3)--(5,3);
    \draw (0,4)--(5,4);
    \draw (1,0)--(1,5);
    \draw (2,0)--(2,5);
    \draw (3,0)--(3,5);
    \draw (4,0)--(4,5);
    \draw[fill=gray!30,opacity=0.5] (1,1) -- (4,1) -- (4,3) -- (1,3) -- cycle;
    \draw[blue, decoration = {markings, mark=at position 0.6 with {\arrow[scale=1.0]{stealth}}}, postaction=decorate] (1,1) -- (4,1);
    \draw[blue, decoration = {markings, mark=at position 0.6 with {\arrow[scale=1.0]{stealth}}}, postaction=decorate] (4,1) -- (4,3);
    \draw[blue, decoration = {markings, mark=at position 0.6 with {\arrow[scale=1.0]{stealth}}}, postaction=decorate] (4,3) -- (1,3);
    \draw[blue, decoration = {markings, mark=at position 0.6 with {\arrow[scale=1.0]{stealth}}}, postaction=decorate] (1,3) -- (1,1);
    \node[blue] at (4.5,2.5) {$L_0$};
    \node[red] at (2.5,2){$D$};
\end{tikzpicture}
    \caption{Inserting $\mathrm{U(1)}^{(0)}_m$ defect along $L_0$ on the lattice.}
    \label{fig:inserting topological defect on the lattice}
\end{figure}
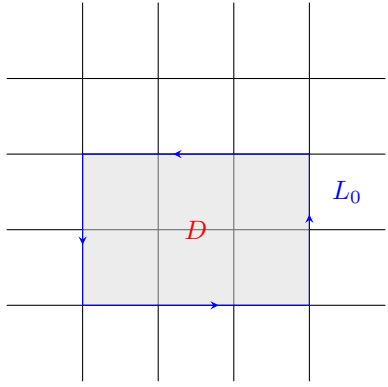

Note that the $2\pi$ periodicity of $\alpha$ is compatible with the identification of $\lambda$ in \eqref{eq:moduli space for U(1) gauge theory}. Thus, \eqref{eq:topolocal defect shifts lambda} is a well-defined action on the moduli space $\mathcal{M}$.

We also note that the modified Gauss law in 1+1d theory with $U(1)$ symmetry was discussed in \cite{Cheng:2022sgb,Seifnashri:2026ema}, where the modification is identified with twisting the boundary condition by $U(1)$ symmetry. This is a lower dimensional and lower symmetry-form analogue of what we discussed above. 

\subsection{Monodromy defects}

However, the action \eqref{eq:topolocal defect shifts lambda} on the first factor of moduli space $\mathcal{M}$ is not transitive, since $\delta\lambda$ is invariant under such shift. This means stating with $\lambda_0$, say the trivial one $\lambda_0=0$, we can only reach other value of $\lambda$ with $\delta \lambda = \delta \lambda_0=0$ by inserting topological defect for $\mathrm{U(1)}^{(0)}_m$ symmetry, and cannot reach the whole space of $\lambda$, i.e. the first factor of $\mathcal{M}$ in \eqref{eq:moduli space for U(1) gauge theory}.

We start with the continuum and go beyond the realm of topological defects a little. Consider defect 
\begin{align}\label{eq:U_x(Sigma)}
    \mathcal{U}_x(\Sigma)=\exp\left(ix\int_\Sigma \frac{F}{2\pi}\right),
\end{align}
where $\Sigma$ is taken to be a surface with non-empty boundary $\gamma$, as in Figure \ref{fig:spacetime open defect}. Now $\int_\Sigma F$ is not quantized on open surface, and hence $x\in\mathbb{R}$. The operator $\mathcal{U}_x(\Sigma)$ is invariant under infinitesimal deformation of the bulk of $\Sigma$, as long as $\gamma=\partial\Sigma$ is unchanged. In short, \eqref{eq:U_x(Sigma)} is a truncated $\mathrm{U(1)}_m^{(0)}$ surface defect.

\begin{figure}
\centering
\begin{tikzpicture}[
    x={(-0.5cm,-0.4cm)},  
    y={(1cm,0cm)},         
    z={(0cm,1cm)},         
    scale=1,
    line join=round,
    line cap=round,
    every node/.style={font=\footnotesize}
]
\draw (2,-2,0) -- (-2,-2,0) -- (-2,2,0);
\begin{scope}[canvas is yz plane at x=0]
    \fill[red!20,opacity=0.6]
        (-1.2,-2) rectangle (1.2,2);
\end{scope}
\draw[blue,decoration = {markings, mark=at position 0.6 with {\arrow[scale=1.0]{stealth}}}, postaction=decorate] (0,-1.2,0)--(0,1.2,0);
\draw[red,decoration = {markings, mark=at position 0.6 with {\arrow[scale=1.0]{stealth}}}, postaction=decorate] (0,-1.2,2)--(0,-1.2,-2);
\draw[red,decoration = {markings, mark=at position 0.6 with {\arrow[scale=1.0]{stealth}}}, postaction=decorate] (0,1.2,-2)--(0,1.2,2);
\draw (2,-2,0) -- (2,2,0) -- (-2,2,0);
\node[blue] at (0.4,0,0) {$L$};
\node[red] at (1,0,2.15) {$\Sigma$};
\node[red] at (0,1.5,1) {$\gamma$};
\end{tikzpicture}
    \caption{The spacetime picture for the defect $\mathcal{U}_x(\Sigma)$, where surface $\Sigma$ with boundary $\gamma$ intersects with the equal time slice on a segment $L$.}
    \label{fig:spacetime open defect}
\end{figure}
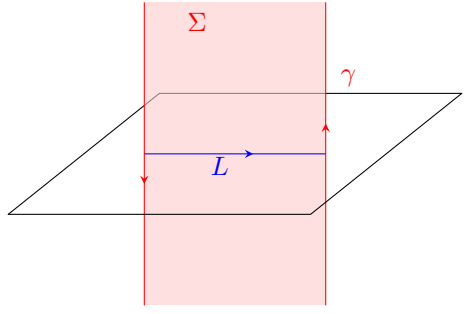

Equivalent, $\mathcal{U}_x(\Sigma)$ is the monodromy defect of the monopole operator \cite{Komargodski:2025jbu,Copetti:2026ncv,Castro:2026mlo,Shao:2025qvf,Barkeshli:2025cjs}. If we insert $\gamma$ at $r=0$ in space and along the time direction, denoting $M_n$ for the charge $n$ monopole operator, the boundary condition in the angular direction is
\begin{align}
    M_n(t,r,\theta+2\pi)=e^{inx}M_n(t,r,\theta),
\end{align}
since the monopole gets acted by the surface operator once it pass through the defect. The surface $\Sigma$ is the branch cut for monodromy, hence topological. However, the boundary of the branch cut is non-topological.\footnote{Since $\mathrm{U(1)}_m^{(0)}$ acts faithfully in the Maxwell theory, its topological surface operator can not be cut open topologically \cite{Aharony:2023amq}.} 

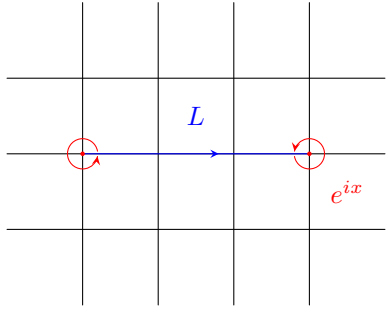
\begin{figure}
\centering
\begin{tikzpicture}[
    scale=1,
    line join=round,
    line cap=round,
    every node/.style={font=\footnotesize}
]
    \draw (0,1)--(5,1);
    \draw (0,2)--(5,2);
    \draw (0,3)--(5,3);
    \draw (1,0)--(1,4);
    \draw (2,0)--(2,4);
    \draw (3,0)--(3,4);
    \draw (4,0)--(4,4);
    \draw[blue, decoration = {markings, mark=at position 0.6 with {\arrow[scale=1.0]{stealth}}}, postaction=decorate] (1,2) -- (4,2);
    \node[blue,thick] at (2.5,2.5) {$L$};
    \node[circle, fill=red, inner sep=0.6pt] at (4,2) {};
    \node[circle, fill=red, inner sep=0.6pt] at (1,2) {};
    \node[red] at (4.5,1.5){$e^{ix}$};
    \begin{scope}[shift={(4,2)}]
        \draw[
            red,
            decoration={markings, mark=at position 1 with {\arrow[scale=1.0]{stealth}}},
            postaction=decorate
        ]
        (-170:0.2) arc[start angle=-170, end angle=170, radius=0.20];
    \end{scope}
    \begin{scope}[shift={(1,2)}]
        \draw[
            red,
            decoration={markings, mark=at position 1 with {\arrow[scale=1.0]{stealth}}},
            postaction=decorate
        ]
        (10:0.2) arc[start angle=10, end angle=350, radius=0.2];
    \end{scope}
\end{tikzpicture}
    \caption{Inserting defect $\mathcal{U}_x(\Sigma)$ along $L$ on the lattice. The monodromy is introduced on the ends of the segment $L$.}
    \label{fig:inserting monodromy defect on the lattice}
\end{figure}

Let $\Sigma$ intersects with the spatial lattice on links $L$ and $\gamma$ intersects on sites $s$, than $L$ is a segment with ends $s=\partial L$, as in Figure \ref{fig:inserting monodromy defect on the lattice}. From the section \ref{sec:topological defect}, we conclude that inserting the defect \eqref{eq:U_x(Sigma)} on the lattice, corresponds to shifting
\begin{align}\label{eq:dynamical defect shifts lambda}
    \lambda_{\ell_*}\to\lambda_{\ell_*}-x \quad \text{for}\quad \ell_*\in L.
\end{align}
This is to be contrasted with \eqref{eq:topolocal defect shifts lambda}, where $L$ are closed loops and $\alpha\in \mathbb{R}/2\pi\mathbb{Z}$. While in \eqref{eq:dynamical defect shifts lambda}, $L$ is an open segment and $x\in\mathbb{R}$. 

One may ask whether the shift \eqref{eq:dynamical defect shifts lambda} can be achieved by conjugating a unitary operator, similar to \eqref{eq:U0twist}. To realize such a shift in the modified Gauss law, one can conjugate it by the unitary
\begin{equation}
    U^{(0)}(x, L) = \exp\left(-i \frac{x}{2\pi} \sum_{\ell \in L} A_\ell\right).
\end{equation}
However, due to the dependence of $A_\ell$, the unitary does not commute with the Hamiltonian. Hence the shift \eqref{eq:dynamical defect shifts lambda} is not achieved by merely conjugating a unitary operator to the system (including the Gauss law and the Hamiltonian).

Note that the action \eqref{eq:dynamical defect shifts lambda} on the space of $\lambda$ is transitive. Start the trivial $\lambda=0$, we can always insert defect $\mathcal{U}_{\lambda_\ell}$ on each link to obtain any nontrivial $\lambda$. The action \eqref{eq:dynamical defect shifts lambda} and \eqref{eq:shifting theta} together is a transitive action on the moduli space $\mathcal{M}$. Hence we have understood that the physical Hilbert space $\mathcal{H}_\chi$, determined by the modified Gauss law $G=\chi$, is the defect Hilbert space with $\mathrm{U(1)}_e^{(1)}$ topological line defects, and monodromy (or twist) defects for $\mathrm{U(1)}_m^{(0)}$. Shifting between presentations in Table \ref{tab:2 presentations} do not change the result.

\section{Magnetic translation}\label{sec:magnetic translation}
In this section, we study the physical Hilbert space $\mathcal{H}_\chi$ by looking at operators on $\mathcal{H}_\chi$, in particular the lattice translation operators. We show that for nontrivial $\chi$, the lattice translations do not commute and satisfy the symmetry algebra for magnetic translations \cite{Zak:1964zz,Hongo:2026rus}, being a manifestation of the insertion monodromy defects concluded in Section \ref{sec:Moduli space and dynamical defects}.

Recall that (physical) operators of theory are those operators acting on the physical Hilbert space $\mathcal{H}_\chi$, i.e. $\mathcal{O}\in\mathrm{End}(\mathcal{H}_\chi)$. Usually we start with an operator $\widetilde{\mathcal{O}}$ on $\mathcal{H}$ and restrict it to the subspace $\mathcal{H}_\chi$, denoted by $\widetilde{\mathcal{O}}|_{\mathcal{H}_\chi}$. We say $\widetilde{\mathcal{O}}$ is physical if
\begin{equation}\label{eq:GOG}
    G(n,\alpha)\widetilde{\mathcal{O}}G^{-1}(n,\alpha)|_{\mathcal{H}_\chi}=\widetilde{\mathcal{O}}|_{\mathcal{H}_\chi}, 
\end{equation}
i.e. $G\widetilde{\mathcal{O}}G^{-1}$ and $\widetilde{\mathcal{O}}$ acts on states in $\mathcal{H}_\chi$ in the same way. Equivalently, since $G^{-1}= \chi^{-1}$ on $\mathcal{H}_\chi$, \eqref{eq:GOG} can also be recast as 
\begin{align}\label{eq:physical operator from Gauss law}
    G(n,\alpha)\widetilde{\mathcal{O}}|_{\mathcal{H}_\chi}=\chi(n,\alpha)\widetilde{\mathcal{O}}|_{\mathcal{H}_\chi}.
\end{align}

Motivated by the results in the last section, we would like to derive the lattice translation operator on $\mathcal{H}_\chi$. We start with the translation on $\mathcal{H}$, which acts on any local operator by
\begin{align}
\begin{split}
    T_x\mathcal{O}_c T^{-1}_x=\mathcal{O}_{c+\hat{x}}, \qquad T_y\mathcal{O}_c T^{-1}_y=\mathcal{O}_{c+\hat{y}},
\end{split}
\end{align}
where $c$ refers to a site, link, or plaquette, and $\hat{x}$, $\hat{y}$ are unite vectors on the square lattice $\Lambda$. The conventions for labeling neighboring sites etc. are listed in Appendix \ref{sec:convensions}.

Let $\chi$ correspond to $(\lambda,\theta)\in\mathcal{M}$. Generally, $\lambda$ and $\theta$ take different values on links and plaquettes, making that $T_x GT_x^{-1}|_{\mathcal{H}_\chi}\neq \chi$. By \eqref{eq:physical operator from Gauss law}, $T_x$ is not a physical operator. However, we propose that a modified translation can be defined as follows, 
\begin{align}\label{eq:modified translation}
\begin{split}
    &T_1=
    T_x\,\exp\left(i\sum_p \lambda_{p+\frac{\hat{x}}{2}}N_p\right)\exp\left(-i\sum_{\ell \parallel \hat{x}}(\delta\lambda)_{\ell+\frac{\hat{x}}{2}}\frac{A_\ell}{2\pi}\right)\exp\left(-i\sum_p\theta_p\left(E+\frac{\delta B}{2\pi}\right)_{p-\frac{\hat{x}}{2}}\right)\\
    &T_2=
    T_y\,\exp\left(-i\sum_p \lambda_{p+\frac{\hat{y}}{2}}N_p\right)\exp\left(-i\sum_{\ell \parallel \hat{y}}(\delta\lambda)_{\ell+\frac{\hat{y}}{2}}\frac{A_\ell}{2\pi}\right)\exp\left(-i\sum_p\theta_p\left(E+\frac{\delta B}{2\pi}\right)_{p-\frac{\hat{y}}{2}}\right)
\end{split}
\end{align}
satisfying \eqref{eq:physical operator from Gauss law}, making it the translation on the physical Hilbert space $\mathcal{H}_\chi$. Here by $\ell\parallel\hat{x}$  we only sum over those links parallel to $x$ direction, similar for $\ell\parallel\hat{y}$. Agian, we work in the first presentation in Table \ref{tab:2 presentations}.

\begin{figure}
\centering
\begin{tikzpicture}[
    scale=1,
    line join=round,
    line cap=round,
    every node/.style={font=\footnotesize}
]

\begin{scope}
    \draw (0,1)--(4,1);
    \draw (0,2)--(4,2);
    \draw (0,3)--(4,3);
    \draw (1,0)--(1,4);
    \draw (2,0)--(2,4);
    \draw (3,0)--(3,4);
    
    \node[circle, fill=red, inner sep=0.6pt] at (1,2) {};
    \node[circle, fill=red, inner sep=0.6pt] at (2,2) {};
    
    \node[blue, fill=white, inner sep=1pt] at (0.5,2) {\footnotesize{$E_1$}};
    \node[blue, fill=white, inner sep=1pt] at (1,2.5) {\footnotesize{$E_2$}};
    \node[blue, fill=white, inner sep=1pt] at (1.5,2) {\footnotesize{$E_3$}};
    \node[blue, fill=white, inner sep=1pt] at (1,1.5) {\footnotesize{$E_4$}};
    \node[blue, fill=white, inner sep=1pt] at (2,2.5) {\footnotesize{$E_5$}};
    \node[blue, fill=white, inner sep=1pt] at (2.5,2) {\footnotesize{$E_6$}};
    \node[blue, fill=white, inner sep=1pt] at (2,1.5) {\footnotesize{$E_7$}};
    
    \node[red] at (1.2,1.75) {\footnotesize{$s_1$}};
    \node[red] at (2.2,1.75) {\footnotesize{$s_2$}};
\end{scope}

\begin{scope}[xshift=5.5cm]
    \draw (0,1)--(4,1);
    \draw (0,2)--(4,2);
    \draw (0,3)--(4,3);
    \draw (1,0)--(1,4);
    \draw (2,0)--(2,4);
    \draw (3,0)--(3,4);
    
    \node[red] at (0.5,2.5) {\footnotesize{$B_1$}};
    \node[red] at (1.5,2.5) {\footnotesize{$B_2$}};
    \node[red] at (2.5,2.5) {\footnotesize{$B_3$}};
    
    \node[blue, fill=white, inner sep=1pt] at (1,2.5) {\footnotesize{$E_1$}};
    \node[blue, fill=white, inner sep=1pt] at (2,2.5) {\footnotesize{$E_2$}};
\end{scope}

\begin{scope}[xshift=11cm]
    \draw (0,1)--(4,1);
    \draw (0,2)--(4,2);
    \draw (0,3)--(4,3);
    \draw (1,0)--(1,4);
    \draw (2,0)--(2,4);
    \draw (3,0)--(3,4);
    
    \node[red] at (1.5,2.5) {\footnotesize{$B_1$}};
    \node[red] at (2.5,2.5) {\footnotesize{$B_3$}};
    \node[red] at (1.5,1.5) {\footnotesize{$B_2$}};
    \node[red] at (2.5,1.5) {\footnotesize{$B_4$}};
    
    \node[blue, fill=white, inner sep=1pt] at (1.5,2) {\footnotesize{$E_1$}};
    \node[blue, fill=white, inner sep=1pt] at (2.5,2) {\footnotesize{$E_2$}};
    \node[blue, fill=white, inner sep=1pt] at (2,2.5) {\footnotesize{$E_3$}};
    \node[blue, fill=white, inner sep=1pt] at (2,1.5) {\footnotesize{$E_4$}};
\end{scope}

\end{tikzpicture}
    \caption{Conjugating the local terms in Gauss law operator with the modified translation operator $T_1$. The three panels shows the label for neighboring links and plaquettes.}
    \label{fig:translating local Gauss law}
\end{figure}
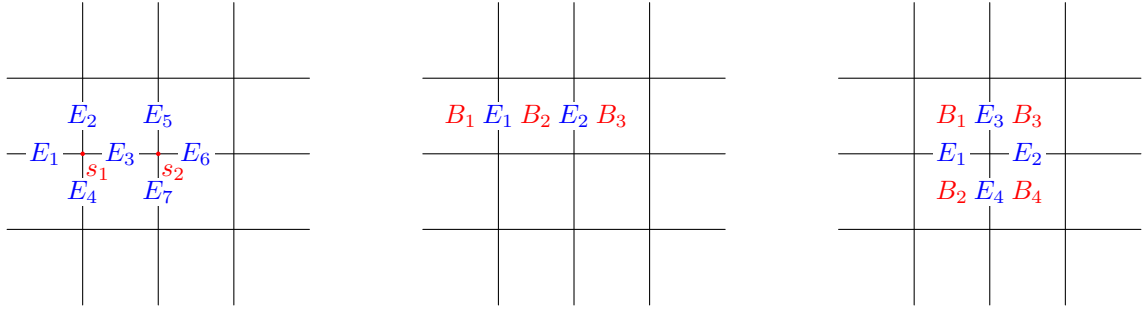

\begin{proof}
In the following, we give the derivation for \eqref{eq:modified translation}.

Since $T_i$ is invertible, we prove \eqref{eq:physical operator from Gauss law} by showing that $T_iGT_i^{-1}|_{\mathcal{H}_\chi}=\chi$. From \eqref{eq:local Gauss law}, we only need to conjugate $T_i$ on each local term in the Gauss law operator $G$. Take $T_1$ as an example. For any site, say $s_1$ in Figure \ref{fig:translating local Gauss law} (left), only the terms involving horizontal neighboring links, i.e. $A_1$ and $A_3$, and $T_x$ conjugate non-trivially on $(\delta E)_{s_1}$,
\begin{align}
\begin{split}
    &T_x \, \exp\biggl(-i(\delta\lambda)_{s_1}\frac{A_1}{2\pi} - i(\delta\lambda)_{s_2}\frac{A_3}{2\pi}\biggr) \, (\delta E)_{s_1} \, \exp\biggl(i(\delta\lambda)_{s_1}\frac{A_1}{2\pi} + i(\delta\lambda)_{s_2}\frac{A_3}{2\pi}\biggr) T_x^{-1} \\
    &\quad = T_x \, (\delta E)_{s_1} \, T_x^{-1} + \frac{(\delta\lambda)_{s_1}}{2\pi} - \frac{(\delta\lambda)_{s_2}}{2\pi} \\
    &\quad = (\delta E)_{s_2} + \frac{(\delta\lambda)_{s_1}}{2\pi} - \frac{(\delta\lambda)_{s_2}}{2\pi}.
\end{split}
\end{align}
Restricting to $\mathcal{H}_\chi$, it equals $\frac{(\delta\lambda)_{s_1}}{2\pi}$, which means that $\delta E$ is invariant when conjugated with $T_1$. For a plaquette $p$, $T_1\exp(2\pi iN_p)T_1^{-1}$ is calculated similarly, and $\exp(2\pi iN_p)$ is invariant restricted on $\mathcal{H}_\chi$.

For any vertical link, say the one on which is $E_1$ in Figure \ref{fig:translating local Gauss law} (middle), only the terms involving neighboring plaquettes, i.e. $N_1$ and $N_2$, and $T_x$ conjugate non-trivially on $e^{i(2\pi E_1+B_1-B_2)}$,
\begin{align}
\begin{split}
    &T_x \, e^{i(\lambda_1 N_1 + \lambda_2 N_2)} \, e^{i(2\pi E_1+B_1-B_2)} \, e^{-i(\lambda_1 N_1 + \lambda_2 N_2)} T_x^{-1} \\
    &\quad = T_x \, e^{i(2\pi E_1+B_1-B_2)} \, T_x^{-1} \, e^{i(\lambda_1-\lambda_2)} \\
    &\quad = e^{i(2\pi E_2+B_2-B_3)} \, e^{i(\lambda_1-\lambda_2)}.
\end{split}
\end{align}
Restricting to $\mathcal{H}_\chi$, it equals $e^{i\lambda_1}$, which means that $e^{i(2\pi E_1+B_1-B_2)}$ is invariant when conjugated with $T_1$.

Finally, for any horizontal link, say the one on which is $E_1$ in Figure \ref{fig:translating local Gauss law} (right), the terms involving $N_1$, $N_2$, $A_1$, and $T_x$ conjugate non-trivially on $e^{i(2\pi E_1+B_1-B_2)}$,
\begin{align}
\begin{split}
    &T_x \, e^{i(\lambda_1 N_1 + \lambda_2 N_2)} \exp\biggl(-i(\delta\lambda)_{s_1}\frac{A_1}{2\pi}\biggr) e^{i(2\pi E_1+B_1-B_2)} \exp\biggl(i(\delta\lambda)_{s_1}\frac{A_1}{2\pi}\biggr) e^{-i(\lambda_1 N_1 + \lambda_2 N_2)} T_x^{-1}\\
    &\quad = T_x \, e^{i(2\pi E_1+B_1-B_2)} \, T_x^{-1} \, e^{i(\lambda_1-\lambda_2)}\\
    &\quad = e^{i(2\pi E_2+B_3-B_4)} \, e^{i(\lambda_3-\lambda_4+(\delta\lambda)_{s_1})}.
\end{split}
\end{align}
Restricting to $\mathcal{H}_\chi$, it equals $e^{i(\lambda_2+\lambda_3-\lambda_4+(\delta\lambda)_{s_1})}=e^{i\lambda_1}$, which means that $e^{i(2\pi E_1+B_1-B_2)}$ is invariant when conjugated with $T_1$.
\end{proof}

In conclusion, we have proven that $T_1GT_1^{-1}|_{\mathcal{H}_\chi}=\chi$, which implies that $T_1$ is an physical operator on $\mathcal{H}_\chi$. Since it is obtained by modifying the translation on $\mathcal{H}$, we conclude that $T_1$ is the $x$ direction translation operator in $\mathcal{H}_\chi$. The proof for $T_2$ is similar.

Note that the translations in two directions do not commute. Motivated by the conclusions in Section \ref{sec:Moduli space and dynamical defects}, i.e. monodromy defect insertions, we calculate 
\begin{align}\label{eq:algebra for magetic translation}
    T_1T_2T_1^{-1}T_2^{-1}=\exp\left[-i\sum_s (\delta\lambda)_s \left(N+\frac{dA}{2\pi}\right)_{s+\frac{\hat{x}}{2}+\frac{\hat{y}}{2}}\right].
\end{align}
Since $\left(N+\frac{dA}{2\pi}\right)_p$ is the gauge invariant local monopole density operator, the magnetic translation algebra \eqref{eq:algebra for magetic translation} is a manifestation of the $e^{i(\delta\lambda)_s}$ monodromy at each site. 

Note that \eqref{eq:algebra for magetic translation} do not depend on the part of $(\lambda,\theta)$ which corresponds to the insertion of \emph{topological} defects, i.e those discussed in Section  \ref{sec:topological defect}, but only depend on the insertion of non-topological monodromy defects. This agrees with the expectation that inserting topological defects alone do not give raise to magnetic translations.

Also, for the Hamiltonian $H$ \eqref{eq:unmodified Hamiltonian} with $e\neq0$, the translation operators $T_i$ \eqref{eq:modified translation} do not commute with the Hamiltonian $H$ \eqref{eq:unmodified Hamiltonian} if $\delta\lambda\neq 0$, meaning that the (magnetic) translation symmetry is \emph{explicitly} broken. Note that the Gauss law with $\delta \lambda \neq 0$ is obtained by inserting monodromy defects. This agrees with the fact that the monodromy defect is not topological, which breaks the translation symmetry. 

\section{Duality to compact boson}\label{sec:duality}

In the continuum, 2+1d free Maxwell theory is dual to free compact boson at radius $R=\frac{e}{2\pi}$,
\begin{align}
    \frac{1}{2e^2}f\wedge\star f \quad \longleftrightarrow \quad \frac{e^2}{8\pi^2}d\phi\wedge\star d\phi.
\end{align}
In this section, we demonstrate this exact duality in the Hamiltonian lattice model, and compare the Gauss laws on both sides. 

The lattice duality was discussed in \cite{Gorantla:2021svj}.
Again we start with the square lattice with periodic boundary conditions, and put a pair of real-valued conjugate variables on each site and link,
\begin{align}
\begin{split}
[\phi_s,p_{s'}]=i\delta_{ss'}, \qquad 
    [\widetilde{\phi}_\ell, w_{\ell'}]=i\delta_{\ell\ell'}.
\end{split}
\end{align}
The $2\pi$ shift of $\phi_s$ is gauged by $w_\ell$, hence the following $\mathbb{Z}$ gauge transformation
\begin{align}
\begin{split}
        \phi \to \phi+2\pi n, \qquad 
    w \to w-dn
\end{split}
\end{align}
for $n\in C^0(\Lambda, \mathbb{Z})$. Except the $2\pi$ periodic\footnote{or the eigenvalues of $w_\ell$ being integers}, $\widetilde{\phi}$ has a more secret gauge redundancy by sifting $\delta \alpha$,
\begin{align}
    \widetilde{\phi}\to\widetilde{\phi}+2\pi\widetilde{n}+\delta\alpha,
\end{align}
for $\widetilde{n}\in C^1(\Lambda,\mathbb{Z})$ and $\alpha\in C^2(\Lambda, \mathbb{R})$. The Gauss law operator can be written down similar to \eqref{eq:U(1) Gauss law operator},
\begin{align}\label{eq:boson Gauss law operator}
    G=\exp\left[i\sum_s n_s\left(2\pi p_s+\delta \widetilde{\phi}_s\right)\right]\exp\left[i\sum_p \alpha_p(dw)_p+i\sum_\ell 2\pi \widetilde{n}_\ell w_\ell\right]
\end{align}
and we first do the unmodified constraints for simplicity,
\begin{align}
    e^{i(2\pi p+\delta \widetilde{\phi})}=1,\quad dw=0,\quad e^{2\pi iw}=1.
\end{align}

The Hamiltonian for the compact boson theory is
\begin{align}\label{eq:ham in compact}
    H=\sum_s \frac{1}{2R^2}p_s^2+\sum_\ell \frac{R^2}{2}(d\phi+2\pi w)_\ell^2.
\end{align}
The  charge operators for 0-from $\mathrm{U(1)}_\text{shift}^{(0)}$ and 1-from $\mathrm{U(1)}_\text{winding}^{(1)}$ symmetries are
\begin{align}
\begin{split}
    Q^{(0)}=\sum_s \left(p+\frac{\delta \widetilde{\phi}}{2\pi}\right)_s, \qquad 
    Q^{(1)}=\sum_\ell w_\ell.
\end{split}
\end{align}
The seasons for these charge operators being topological, and the local charges being integer-valued are similarly to the discussions in Section \ref{sec:symmetries}.

Now we introduce a new set of variables to describe this theory, defined by
\begin{align}
\begin{split}
    &A_\ell = \widetilde{\phi}_\ell,\quad E_\ell=w_\ell+\left(\frac{d\phi}{2\pi}\right)_\ell,\\
    &B_s = \phi_s, \quad  N_s=p_s+\left(\frac{\delta \widetilde{\phi}}{2\pi}\right)_s.
\end{split}
\end{align}
If we relabel variables by the dual lattice\footnote{ namely exchange sites with plaquettes, $d$ with $\delta$, etc..}, it can be checked that these variables satisfy the commutation relations \eqref{eq:abelian gauge theory commutator 1} \eqref{eq:abelian gauge theory commutator 2} in the abelian gauge theory, the Hamiltonian of the boson theory \eqref{eq:ham in compact} takes exactly the same form as \eqref{eq:unmodified Hamiltonian} with $e=2\pi R$, so are the 0-from and 1-from charge operators. In particular, we focus on the constraints
\begin{align}\label{eq:duality between constraints}
\begin{split}
    e^{i(2\pi p+\delta \widetilde{\phi})}=1 \quad &\longleftrightarrow \quad e^{2\pi iN}=1\\
    dw=0\quad &\longleftrightarrow \quad dE=0\\
    e^{2\pi iw}=1 \quad&\longleftrightarrow \quad e^{i(2\pi E+d B)}=1,
\end{split}
\end{align}
which recovers the unmodified constraints in the abelian gauge theory \eqref{eq:trandidtional Gauss law}.

We make two comments here. Firstly, the compact boson in the continuum is not a gauge theory in an obvious way, but the corresponding lattice model has constraints, that are matched with the lattice gauge theory. Secondly, in the compact boson presentation, the constraint $e^{2\pi iw}=1$ which fixes the eigenvalues of $w_\ell$ to be integers, gets mapped to the Gauss law in the $\mathrm{U(1)}$ gauge theory presentation, as shown in \eqref{eq:duality between constraints}. So the gauge transformation associated with the gauge group is not an unambitious notion, and we suggest just regard them as constraints in the Hamiltonian lattice model.

Of course, we can work out all possible modified Gauss laws for the Gauss law operator \eqref{eq:boson Gauss law operator} in the compact boson theory. The calculation and the result, i.e. the moduli space $\mathcal{M}$, are basically the same as Section \ref{sec:Moduli space of physical Hilbert spaces}, matching with the $\mathrm{U(1)}$ gauge theory presentation.

\section{Discussions and future directions}\label{sec:Discussions and future directions}
One generalization is that we couple the $\mathrm{U(1)}$ gauge field to matter degrees of freedom, making it an interacting theory. A subtlety in the lattice is that if the eigenvalues of the local charge operator of the matter is finite, such as $\rho_s=\sigma^z_s$ for a qubit whose eigenvalues are $\pm 1$, the Gauss law with some choice of $\chi$ could lead to non-extensive or empty physical Hilbert space.\footnote{This subtlety has already been noticed in Kogut and Susskind \cite{Kogut:1979wt,Kogut:1974ag} and recently in \cite{Chatterjee:2024gje}.} Hence the physical interpretation of the Gauss law is more subtle.

If we couple to a complex scalar field where the local Hilbert space is infinite dimensional, for example considering the abelian-Higgs model, than such problem does not exist and we can also think about the modified Gauss laws. Then part of the Gauss law operator changes
\begin{align}
    \exp\left(i\sum_s \alpha_s(\delta E)_s\right) \to  \exp\left(i\sum_s \alpha_s(\delta E+\rho)_s\right).
\end{align}
Some consistency conditions used in Section \ref{sec:Moduli space of physical Hilbert spaces} no longer hold, hence the moduli space we discussed does not directly apply to the abelian-Higgs model. Despite this, we can still discuss the interpretation of the modification \eqref{eq:local Gauss law} or \eqref{eq:modified gauge theory Hamiltonian}.

Note that introducing the scalar does not affect the third modified Gauss law in \eqref{eq:local Gauss law}, i.e. $\exp(2\pi iN-i\theta)=1$. In the abelian-Higgs model this cannot correspond to the insertion of $\mathrm{U(1)}_e^{(1)}$ line defect since the $\mathrm{U(1)}_e^{(1)}$ symmetry is explicitly broken by the charge 1 complex scalar, but the insertion of vortex loop defect:
\begin{align}
    \lim_{r\to0}\exp\left(i\int_{C_r}a\right)=e^{i\theta}.
\end{align}
If we shift to the modified Hamiltonian presentation in Tabel \ref{tab:2 presentations}, this has been discussed in \cite{Komargodski:2025jbu}. Similar phenomenon also occurs in finite lattice gauge theories, see Appendix \ref{app:finite}.\footnote{In the finite lattice gauge theory, adding matter degrees of freedom can be achieved by only modifying the Hamiltonian, without introducing degrees of freedom on sites. See Appendix \ref{app:finite} for details. } 

On the other hand, the magnetic $\mathrm{U(1)}^{(0)}_m$ symmetry still exists. So we claim that the modification of $\lambda$ still corresponds to the insertion of monodromy defect in the interacting theory.

Other future directions include analyzing the modified Gauss laws for nonabelian gauge theories on the lattice \cite{Chen:2024ddr}, possible realization of the anomalous translation \cite{Seiberg:2024wgj,Seiberg:2024yig} on the lattice, and modifying the non-invertible Gauss law \cite{Choi:2022fgx}.

\section*{Acknowledgments}
We are grateful to Yichul Choi, Shu-Heng Shao for helpful discussions, and  Sakura Schafer-Nameki, Sahand Seifnashri for comments on a draft. The work of Y.Z.
is supported by NSFC grant No.12505093 and the starting funds from University of Chinese Academy of Sciences
(UCAS) and from the Kavli Institute for Theoretical Sciences (KITS).

\appendix
\section{Conventions on (co)chains on square lattice}\label{sec:convensions}
We consider 2d square lattice $\Lambda$ with periodic boundary conditions in both direction, where $s\in C_0(\Lambda,\mathbb{Z})$ denotes sites, $\ell\in C_1(\Lambda,\mathbb{Z})$ denotes links and $p\in C_1(\Lambda,\mathbb{Z})$ denotes plaquettes. The unit vectors in positive $x$ and $y$ directions are denoted as $\hat{x}$ and $\hat{y}$. The links are oriented in the positive $x$ or $y$ direction. The direction reversed for $\ell$ is denoted as $-\ell$.

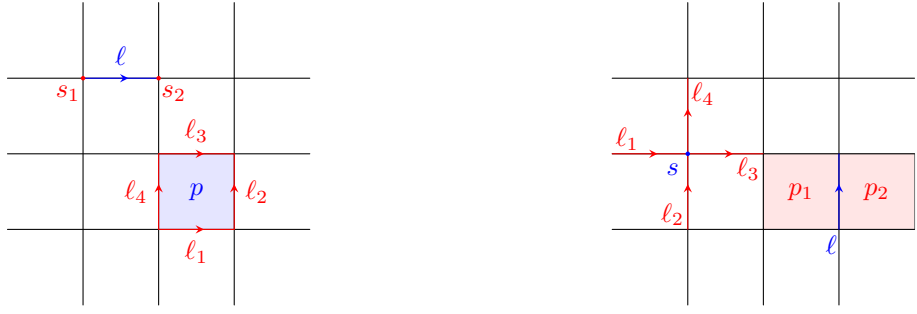
\begin{figure}
\centering
\begin{tikzpicture}[
    scale=1,
    line join=round,
    line cap=round,
    every node/.style={font=\footnotesize}
]
\begin{scope}
    \draw (0,1)--(4,1);
    \draw (0,2)--(4,2);
    \draw (0,3)--(4,3);
    \draw (1,0)--(1,4);
    \draw (2,0)--(2,4);
    \draw (3,0)--(3,4);
    \draw[blue, decoration = {markings, mark=at position 0.6 with {\arrow[scale=1.0]{stealth}}}, postaction=decorate] (1,3) -- (2,3);
    \node[circle, fill=red, inner sep=0.6pt] at (1,3) {};
    \node[circle, fill=red, inner sep=0.6pt] at (2,3) {};
    \draw[fill=blue!20,opacity=0.5] (2,1) -- (3,1) -- (3,2) -- (2,2) -- cycle;
    \draw[red, decoration = {markings, mark=at position 0.6 with {\arrow[scale=1.0]{stealth}}}, postaction=decorate] (2,1) -- (3,1);
    \draw[red, decoration = {markings, mark=at position 0.6 with {\arrow[scale=1.0]{stealth}}}, postaction=decorate] (2,2) -- (3,2);
    \draw[red, decoration = {markings, mark=at position 0.6 with {\arrow[scale=1.0]{stealth}}}, postaction=decorate] (2,1) -- (2,2);
    \draw[red, decoration = {markings, mark=at position 0.6 with {\arrow[scale=1.0]{stealth}}}, postaction=decorate] (3,1) -- (3,2);
    \node[blue] at (1.5,3.3) {\footnotesize{$\ell$}};
    \node[red] at (0.8,2.8) {\footnotesize{$s_1$}};
    \node[red] at (2.2,2.8) {\footnotesize{$s_2$}};
    \node[blue] at (2.5,1.5) {\footnotesize{$p$}};
    \node[red] at (2.5,0.7) {\footnotesize{$\ell_1$}};
    \node[red] at (2.5,2.3) {\footnotesize{$\ell_3$}};
    \node[red] at (1.7,1.5) {\footnotesize{$\ell_4$}};
    \node[red] at (3.3,1.5) {\footnotesize{$\ell_2$}};
\end{scope}
\begin{scope} [xshift=8cm]
    \draw (0,1)--(4,1);
    \draw (0,2)--(4,2);
    \draw (0,3)--(4,3);
    \draw (1,0)--(1,4);
    \draw (2,0)--(2,4);
    \draw (3,0)--(3,4);
    \draw[red, decoration = {markings, mark=at position 0.6 with {\arrow[scale=1.0]{stealth}}}, postaction=decorate] (0,2) -- (1,2);
    \draw[red, decoration = {markings, mark=at position 0.6 with {\arrow[scale=1.0]{stealth}}}, postaction=decorate] (1,2) -- (1,3);
    \draw[red, decoration = {markings, mark=at position 0.6 with {\arrow[scale=1.0]{stealth}}}, postaction=decorate] (1,2) -- (2,2);
    \draw[red, decoration = {markings, mark=at position 0.6 with {\arrow[scale=1.0]{stealth}}}, postaction=decorate] (1,1) -- (1,2);
    \node[circle, fill=blue, inner sep=0.6pt] at (1,2) {};
    \draw[fill=red!20,opacity=0.5] (2,1) -- (3,1) -- (3,2) -- (2,2) -- cycle;
    \draw[fill=red!20,opacity=0.5] (3,1) -- (4,1) -- (4,2) -- (3,2) -- cycle;
    \draw[blue, decoration = {markings, mark=at position 0.6 with {\arrow[scale=1.0]{stealth}}}, postaction=decorate] (3,1) -- (3,2);
    \node[red] at (2.5,1.5) {\footnotesize{$p_1$}};
    \node[red] at (3.5,1.5) {\footnotesize{$p_2$}};
    \node[blue] at (0.8,1.8) {\footnotesize{$s$}};
    \node[blue] at (2.9,0.8) {\footnotesize{$\ell$}};
    \node[red] at (0.2,2.2) {\footnotesize{$\ell_1$}};
    \node[red] at (0.8,1.2) {\footnotesize{$\ell_2$}};
    \node[red] at (1.8,1.8) {\footnotesize{$\ell_3$}};
    \node[red] at (1.2,2.8) {\footnotesize{$\ell_4$}};
\end{scope}
\end{tikzpicture}
    \caption{Left: the definition of operators $\partial$ and $d$. Right: the definition of operators $\hat{\delta}$ and $\delta$.}
    \label{fig:definition of bdy and its adjoint}
\end{figure}

The boundary operator is defined as $\partial: C_k(\Lambda,\mathbb{Z})\to C_{k-1}(\Lambda,\mathbb{Z})$. In particular,
\begin{align}
    \partial \ell=s_2-s_1,\quad \text{and}\quad \partial p=\ell_1+\ell_2-\ell_3-\ell_4,
\end{align}
in Figure \ref{fig:definition of bdy and its adjoint}. 

Let $A$ be an abelian group, like $\mathbb{R}$ or $\mathbb{Z}$. The coboundary operator is defined as $d: C^k(\Lambda,A)\to C^{k+1}(\Lambda, A)$, by 
\begin{equation}
    (da)_{c}=a_{\partial c}.
\end{equation}
In particular 
\begin{align}
    (d\phi)_\ell=\phi_{s_2}-\phi_{s_1},\quad \text{and}\quad (da)_p=a_{l_1}+a_{l_2}-a_{l_3}-a_{l_4},
\end{align}
for $\phi\in C^0(\Lambda,A)$ and $a\in C^1(\Lambda,A)$ in Figure \ref{fig:definition of bdy and its adjoint}. 

We also use the operator $\hat{\delta}: C_k(\Lambda,\mathbb{Z})\to C_{k+1}(\Lambda,\mathbb{Z})$ defined by
\begin{align}
    \hat{\delta} c =\sum_{f\,\text{such that }c\subset \partial f}f.
\end{align}
As in Figure \ref{fig:definition of bdy and its adjoint},
\begin{align}
    \hat{\delta}s=\ell_1+\ell_2-\ell_3-\ell_4,\quad \text{and},\quad\hat{\delta}\ell=p_1-p_2.
\end{align}

The corresponding operator acting on cochains $\delta: C^k(\Lambda,A)\to C^{k-1}(\Lambda,A)$, is defined by 
\begin{align}\label{eq:geometric  def for delta}
    (\delta a)_{c}=a_{\hat{\delta}c}.
\end{align}
As in Figure \ref{fig:definition of bdy and its adjoint},
\begin{align}
    (\delta a)_s=a_{\ell_1}+a_{\ell_2}-a_{\ell_3}-a_{\ell_4},\quad\text{and}\quad (\delta b)_\ell=b_{p_1}-b_{p_2}
\end{align}

Actually, $\delta$ has an equivalent definition. View $C^k(\Lambda,\mathbb{R})$ as a real vector space, with the inner product 
\begin{align}\label{eq:inner product of cochain}
    \langle a,b\rangle=\sum_{c\in C_k(\Lambda)} a_cb_c. 
\end{align}
Since $d$ is a linear map, we define $\delta$ as the adjoint of $d$ with respect to such inner product, i.e.
\begin{align}\label{eq:adjoint definition for delta}
    \langle a,db\rangle=\langle\delta a,b\rangle.
\end{align}
It can be shown by direct calculation that this definition is equivalent to the geometric definition \eqref{eq:geometric  def for delta}. \eqref{eq:adjoint definition for delta} can be thought of as the lattice version of integral by part and is used in the paper.

We also use the following notation for refer to neighboring links or plaquettes in Section \ref{sec:magnetic translation}.

\begin{align}
\begin{array}{@{}c@{\hspace{2.5cm}}c@{}}
    \begin{tikzpicture}[inline tikz]
        \draw[thick, mid arrow] (0,0) -- (1.5,0);
        \fill[red] (1.5,0) circle (1.5pt);
        \node[below] at (0.75,0) {$\ell$};
        \node[above right, red, inner sep=2pt] at (1.5,0)
            {$\ell + \frac{\hat{x}}{2}$};
    \end{tikzpicture}
    &
    \begin{tikzpicture}[inline tikz]
        \draw[thick] (1,0) -- (0,0) -- (0,1) -- (1,1);
        \draw[thick, red, mid arrow] (1,0) -- (1,1);
        \node at (0.5,0.5) {$p$};
        \node[right, red] at (1,0.5)
            {$p + \frac{\hat{x}}{2}$};
    \end{tikzpicture}
    \\[1cm]
    \begin{tikzpicture}[inline tikz]
        \draw[thick, mid arrow] (0,-0.75) -- (0,0.75);
        \fill[red] (0,0.75) circle (1.5pt);
        \node[right] at (0,0) {$\ell$};
        \node[above right, red, inner sep=2pt] at (0,0.75)
            {$\ell + \frac{\hat{y}}{2}$};
    \end{tikzpicture}
    &
    \begin{tikzpicture}[inline tikz]
        \draw[thick] (0,1) -- (0,0) -- (1,0) -- (1,1);
        \draw[thick, red, mid arrow] (0,1) -- (1,1);
        \node at (0.5,0.5) {$p$};
        \node[above, red] at (0.5,1)
            {$p + \frac{\hat{y}}{2}$};
    \end{tikzpicture}
    \\[1cm]
    \multicolumn{2}{c}{
    \begin{tikzpicture}[inline tikz]
        \draw[thick] (0,1) -- (0,0) -- (1,0) -- (1,1) -- (0,1);
        \node[red] at (0.5,0.5) {$p$};
        \node[right, red, xshift=4pt, inner sep=2pt] at (1,0.5)
            {$p=s+\frac{\hat{x}}{2}+\frac{\hat{y}}{2}$};
        \node at (-0.2,-0.2) {$s$};
    \end{tikzpicture}
    }
\end{array}
\end{align}

\section{Proof of \eqref{eq:existing lambda for chi}: constructing $\lambda$ from $\chi(n,\alpha,0)$}\label{sec:proof of lambda for chi}

We prove \eqref{eq:existing lambda for chi} in the appendix, i.e. there exist $\lambda\in C^1(\Lambda,\mathbb{R})$ such that 
\begin{align}
    q=\frac{\delta \lambda}{2\pi},\quad \text{and} \quad \varphi=\lambda\mod 2\pi.
\end{align}

To start, we need to prove a lemma which is used in the following.
\begin{lemma}\label{lemma}
Let $\Lambda$ be a square lattice with periodic boundary conditions in both directions. For any $m\in C^0(\Lambda,\mathbb{Z})$ satisfying $\sum_s m_s=0$, there exists $j\in C^1(\Lambda,\mathbb{Z})$, such that $m=\delta j$.
\end{lemma}
\begin{proof}
We provide an explicit construction of $j$ from $m$.

For every site $s\neq (0,0)$, we pick a nearest-neighbor path $\gamma_s$ from the origin $(0,0)$ to $s$. For example, if $s=(a,b)$, take the path that first move horizontally from $(0,0)$ to $(a,0)$, than vertically from $(a,0)$ to $(a,b)$.

Define $j$ by putting $m_s$ along the path $\gamma_s$, for every $s\neq (0,0)$,
\begin{align}\label{eq:j in the lemma 1}
    j=\sum_{s\neq (0,0)} m_s1_{\gamma_s}.
\end{align}
Here $1_{\gamma_s}$ is an integer 1-cochain which is $+1$ on the links of $\gamma_{s}$ with positive orientation, $-1$ on edges taken the negative orientation, and $0$ elsewhere.

Note that 
\begin{align}
    \delta 1_{\gamma_s}=1_s-1_{(0,0)},
\end{align}
where $1_s$ is an integer 0-cochain which is $1$ on site $s$ and $0$ elsewhere. So
\begin{align}
\begin{split}
    \delta j&=\sum_{s\neq (0,0)} m_s \left(1_s-1_{(0,0)}\right)\\
    &=\sum_{s\neq (0,0)} m_s 1_s-\left(\sum_{s\neq (0,0)}m_s\right) 1_{(0,0)}\\
    &=\sum_{s\neq (0,0)} m_s 1_s+m_{(0,0)}1_{(0,0)}\\
    &=\sum_{s} m_s 1_s \equiv m,
\end{split}
\end{align}
where the third equality follows from $\sum_s m_s=0$. Thus, $j$ defined in \eqref{eq:j in the lemma 1} is the one stated in the Lemma \ref{lemma}.

\end{proof}

Now we can prove \eqref{eq:existing lambda for chi}. As mentioned in the remarks below \eqref{eq:U(1) physical Hilbert space}, $\chi(dk,-2\pi k,0)=G(dk,-2\pi k,0)=1$, for $k\in C^0(\Lambda,\mathbb{Z})$. The second equality follows from the definition of Gauss law operator $G$ \eqref{eq:U(1) Gauss law operator}. Thus, 
\begin{align}
\begin{split}
    1&=\chi(dk,-2\pi k,0)=\chi(dk,0,0)\chi(0,-2\pi k,0)\\
    &=\exp\left(i\sum_\ell (dk)_\ell \varphi_\ell-2\pi i\sum_s k_s q_s \right)\\
    &=\exp\left(i\sum_s k_s(\delta\varphi-2\pi q)_s\right),
\end{split}
\end{align}
where the second equality follows from $\chi$ being a homomorphism, and the last equality follows from the property \eqref{eq:adjoint definition for delta}. Since this holds for any $k_s\in\mathbb{Z}$, 
\begin{align}\label{eq:relation varphi and q}
    (\delta \varphi)_s=2\pi q_s \mod 2\pi.
\end{align}

Recall that by definition $\varphi \in C^1(\Lambda,\mathbb{R}/2\pi\mathbb{Z})$, we pick an arbitrary lift from $\mathbb{R}/2\pi\mathbb{Z}$ to $\mathbb{R}$, i.e. $\widetilde{\varphi}\in C^1(\Lambda,\mathbb{R})$,
\begin{align}
    \widetilde{\varphi}_\ell=\varphi_\ell \mod 2\pi.
\end{align}
Then, by \eqref{eq:relation varphi and q} there exist integer $m_s\in\mathbb{Z}$ such that 
\begin{align}\label{eq:relation m, q, and varphi tilde}
    \delta\widetilde{\varphi}-2\pi q=2\pi m.
\end{align}
Note that this equality holds on $\mathbb{R}$ since $q_s\in\mathbb{R}$.

Consider a configuration of $\alpha\in C^0(\Lambda,\mathbb{R})$ such the $\alpha_s$ takes a constant value on every site, $\alpha_s=x$. Since $\sum_s(\delta E)_s=0$ on closed lattice (periodic boundary condition),
\begin{align}
    \exp\left(ix\sum_sq_s\right)=\chi(0,\alpha,0)=G(0,\alpha,0)=\exp\left(ix\sum_s (\delta E)_s\right)=1,
\end{align}
imply that $\sum_s q_s=0$. Thus, from \eqref{eq:relation m, q, and varphi tilde} $m\in C^0(\Lambda,\mathbb{Z})$ satisfies
\begin{align}
    \sum_sm_s=0
\end{align}

From Lemma \ref{lemma}, there exists $j\in C^1(\Lambda,\mathbb{Z})$ such that $m=\delta j$, as constructed in the proof. Thus, from \eqref{eq:relation m, q, and varphi tilde},
\begin{align}
    2\pi q=\delta(\widetilde{\varphi}-2\pi j).
\end{align}
We define
\begin{align}
    \lambda=\widetilde{\varphi}-2\pi j\in C^1(\Lambda,\mathbb{R}),
\end{align}
hence \eqref{eq:existing lambda for chi} is proven.

\section{Compare with lattice finite gauge theories}
\label{app:finite}

In this appendix, we revisit the Gauss laws in $\mathbb{Z}_2$ lattice models in 2+1d lattice gauge theories, following \cite{Seiberg:2024gek,Choi:2024rjm}, and compare them with the Gauss laws in the $\mathrm{U(1)}$ lattice gauge theory.

We consider a square lattice, with spin-$\frac12$ gauge field degrees of freedom, and the local Hilbert space is $\mathcal{H}_\ell = \mathbb{C}^2$, acted upon by $\sigma^a_\ell$, $a=x,y,z$ and $\ell$ labels the link. Similar to the $\mathrm{U(1)}$ examples in the main text, we don't include matter degrees of freedom on sites.

The modified Gauss law is 
\begin{equation}\label{eq:Z2Gauss}
    \prod_{s\in \partial \ell} \sigma^x_{\ell} = (-1)^{\alpha_s}.
\end{equation}
A Hamiltonian commuting with \eqref{eq:Z2Gauss} is 
\begin{equation}\label{eq:Z2Ham}
    H=  g_1  \sum_{p} \prod_{\ell\in \partial p} \sigma^z_\ell + g_2 \sum_\ell \sigma^x_\ell.
\end{equation}
This is the $\bZ_2$ lattice gauge theory studied in \cite{Fradkin:1978dv,Verresen:2022mcr,Dumitrescu:2026vre,Ji:2026yfj}.\footnote{These papers impose the Gauss law constraint energetically  and hence allows an additional interaction $\sum_{\ell} \sigma^z_\ell$. In our work, we always enforce the Gauss law strictly.  }

The Hamiltonian \eqref{eq:Z2Ham}, for any value of $g_1, g_2$, has a $\bZ_2^m$ 1-form symmetry, generated by the unitary line operator 
\begin{equation}\label{eq:Z2sym}
    U(L^*)= \prod_{\ell\in L^*} \sigma^x_{\ell},
\end{equation}
where $L^*$ is a closed loop in the dual lattice. When $g_2=0$, the Hamiltonian \eqref{eq:Z2Ham} has an additional $\bZ_2^e$ 1-form symmetry, generated by 
\begin{equation}
    V(L)= \prod_{\ell \in L} \sigma^z_\ell
\end{equation}
where $L$ is a closed loop in the square lattice.

The modified Gauss law \eqref{eq:Z2Gauss} can be mapped to the ordinary Gauss law by conjugating the truncated $V$ operator 
\begin{equation}\label{eq:Vbdy}
    V(L,s) = \prod_{\ell \in L, s\in \partial L} \sigma^z_\ell.
\end{equation}
When $g_2=0$, since $V$ is a $\bZ_2^e$ symmetry generator, conjugation by \eqref{eq:Vbdy} amounts to inserting a $\bZ_2^e$ 1-form symmetry defect. Hence the modified Gauss law can be obtained from the ordinary Gauss law by inserting a $\bZ_2^e$ 1-form symmetry defect.

When $g_2\neq 0$, since $V$ is no longer a symmetry, conjugating \eqref{eq:Vbdy} with \eqref{eq:Z2Ham} flips the sign of $\sigma^x_\ell$ along $L$, hence the above interpretation fails. 
To see the physical meaning of the modified Gauss law, let's consider the continuum limit. The Hamiltonian \eqref{eq:Z2Ham} flows to a real scalar $\phi$ coupling to a dynamical $\bZ_2$ gauge field $a$. The Lagrangian is 
\begin{equation}
    (D_{\pi a}\phi)^2 + m^2 \phi^2 + u \phi^4 
\end{equation}
The $\sigma^x_\ell$ on each link is regarded as a Wilson line segment $(-1)^{a_\ell}$ of the $\bZ_2$ gauge field. The ordinary Gauss law $\prod_{s\in \partial \ell} \sigma^x_{\ell} =1$ ensures that the gauge field is flat $(-1)^{da}=1$. Suppose the Gauss law is modified at one site $s$. It means that the flatness of the $\bZ_2$ gauge field is violated at one plaquette $p=s^*$ in the dual lattice. Around that plaquette, the $\bZ_2$ gauge field has $\bZ_2$-flux, $(da)_{s^*}=1$, or equivalently, the $\bZ_2$ Wilson loop has a non-trivial phase 
\begin{equation}
    \exp\left( i \pi \oint_{\gamma} a\right) = -1, \qquad \text{link}(\gamma,s)=1 
\end{equation}
This defect is the $\bZ_2$-vortex, as studied in \cite{Komargodski:2025jbu}. In summary, we find that modification of the Gauss law \eqref{eq:Z2Gauss} is equivalent to inserting $\bZ_2$-vortices.

\bibliographystyle{JHEP}
\bibliography{ref}

\end{document}